\pdfoutput=1
\documentclass{article}
\usepackage[noblocks]{authblk}
\let\oldparagraph\paragraph
\renewcommand{\paragraph}[1]{\oldparagraph{#1.}}
\usepackage{amsmath}
\usepackage{amsthm}
\usepackage{amssymb}
\usepackage[backgroundcolor=yellow,colorinlistoftodos]{todonotes}
\usepackage{xcolor}
\usepackage{xspace}
\usepackage{setspace}
\usepackage{multirow}
\usepackage{color,svgcolor}
\usepackage{paralist}
\usepackage{booktabs}

\usepackage{tikz}
\usetikzlibrary{matrix}
\usepackage{array} %
\newcolumntype{L}{>{$}p{1.9cm}<{$\hfil}} %

\usepackage{arydshln} %

\makeatletter
\def\mathcolor#1#{\@mathcolor{#1}}
\def\@mathcolor#1#2#3{%
\protect\leavevmode
\begingroup
\color#1{#2}#3%
\endgroup
}
\makeatother
\def\triadone{brown}
\def\triadtwo{red}
\def\triadthree{blue}
\def\triadfour{green}
\def\triadfive{purple}
\def\triadsix{dimgray}
\def\triadseven{chocolate}
\def\matrixsize#1#2{{{#1}\times{#2}}}
\def\tensorproduct{\otimes}
\def\tensor#1{{\mathcal{#1}}}

\def\Trace{\textup{Trace}}
\def\Inverse#1{{#1}^{-1}}
\def\MatrixProduct#1#2{{\mat{#1}\cdot\mat{#2}}}

\DeclareMathOperator{\Identity}{I}
\def\IdentityMatrix{\Identity}
\def\IdentityMatrixOfSize#1{{\IdentityMatrix}_{#1}}
\newenvironment{smatrix}{\left(\begin{smallmatrix}}{\end{smallmatrix}\right)}
\newenvironment{mmatrix}{\left(\begin{matrix}}{\end{matrix}\right)}
\newcommand{\Ring}{\ensuremath{\mathfrak{R}}}
\newcommand{\Sring}{\ensuremath{\mathfrak{S}}}
\newcommand{\Matrices}[2]{\ensuremath{{{#1}^{#2}}}}
\newcommand{\BilinearMapCoDomain}{\ensuremath{\mathfrak{M}}}
\newcommand{\BilinearMapCoDomainSubSpace}{\ensuremath{\mathcal{L}}}
\newcommand{\Image}[1]{\ensuremath{{\textrm{Im}\,#1}}}

\newcommand{\Dual}[1]{\ensuremath{{(#1)^{\star}}}}
\newcommand{\Contraction}[3]{\ensuremath{{{#1}\!\mid^{#2}_{#3}}}}
\newcommand{\Rank}[1]{\ensuremath{{\textup{Rank}\,{#1}}}}

\newcommand{\mat}[1]{\ensuremath{{#1}}}%
\newcommand{\Transpose}[1]{{{\mat{#1}}^{\intercal}}\xspace}
\newcommand{\InvTranspose}[1]{{({#1}^{-1})^{\intercal}}}
\newcommand{\CTranspose}[1]{{{\overline{\mat{#1}}}^{\intercal}}\xspace}
\newcommand{\HermitianTranspose}[1]{\ensuremath{{{\mat{#1}}^{H}}}}
\newcommand{\Conjugation}[1]{\ensuremath{{\overline{#1}}}}
\newcommand{\Conjugate}[1]{\ensuremath{{\Conjugation{#1}}}}
\newcommand{\SymbolIAH}{\ensuremath{\phi}}
\newcommand{\MIAH}[1]{\ensuremath{{\SymbolIAH\left({\mat{#1}}\right)}}}
\newcommand{\LLTr}{\ensuremath{{\text{Low}}}}
\newcommand{\UpTr}{\ensuremath{{\text{Up}}}}
\newcommand{\Low}[1]{\ensuremath{{\LLTr\left({\mat{#1}}\right)}}}
\newcommand{\Up}[1]{\ensuremath{{\UpTr\left({\mat{#1}}\right)}}}

\newcommand{\mathsc}[1]{{\normalfont\textsc{#1}}}
\newcommand{\Primes}{\ensuremath{\mathbb{P}}\xspace}
\newcommand{\Z}{\ensuremath{\mathbb{Z}}\xspace}
\newcommand{\N}{\ensuremath{\mathbb{N}}\xspace}
\newcommand{\K}{\ensuremath{\mathbb{K}}\xspace}
\newcommand{\E}{\ensuremath{\mathbb{E}}\xspace}
\newcommand{\HK}[1][\K]{\ensuremath{\mathbb{H}(#1)}\xspace}
\newcommand{\F}{\ensuremath{\mathbb{F}}\xspace}
\newcommand{\RR}{\ensuremath{\mathbb{R}}\xspace}

\newcommand{\QQ}{\ensuremath{\mathbb{Q}}\xspace}
\newcommand{\CC}{\ensuremath{\mathbb{C}}\xspace}

\newcommand{\quat}[1]{\ensuremath{\mathbf{#1}}}

\newcommand{\FFSoS}[2]{\ensuremath{\text{\texttt{SoS}}(#1,#2)}}
\newcommand{\FFSqrtname}{\ensuremath{\text{\texttt{sqrt}}}}
\newcommand{\FFSqrt}[2]{\ensuremath{\FFSqrtname(#2)}}

\newcommand{\GO}[1]{\ensuremath{{O\mathopen{}\left({#1}\right)\mathclose{}}}\xspace}
\newcommand{\LO}[1]{\ensuremath{{o\mathopen{}\left({#1}\right)\mathclose{}}}\xspace}
\newcommand{\assign}{\ensuremath{\leftarrow\xspace}}
\newcommand{\MM}[2][vide]{%
  \ifthenelse{ \equal{#1}{vide} }
  {\ensuremath{\textrm{MM}_{#2}}}
  {\ensuremath{\textrm{MM}^{#1}_{#2}}}
}

\def\SymFirstTriad{\alpha}
\def\SymSecondTriad{\beta}
\def\SymThirdTriad{\gamma}
\def\Isotropy#1{\mathsf{#1}}
\def\IsotropyAction#1#2{{{#1}\diamond{#2}}}
\newcommand{\Subs}[2]{\ensuremath{{{#2}_{#1}}}}
\newcommand{\Minor}[3]{\ensuremath{{\textrm{Det}_{[{#1}|{#2}]}({#3})}}}
\newcommand{\AddRow}[4]{\ensuremath{{#4}|{[{#1},{#2}|{#3}]}}}
\usepackage{hyperref}
\makeatletter
\hypersetup{%
	breaklinks=true,
	colorlinks=true,
	linkcolor=darkred,
	citecolor=blue,
	urlcolor=darkgreen,
}
\makeatother
\usepackage[nameinlink,capitalize,noabbrev]{cleveref}

\usepackage{caption}
\usepackage{algorithm}
\usepackage[noend]{algpseudocode} %
\algnewcommand{\algorithmicassumption}{\textbf{Requirement:}}
\algnewcommand{\Assume}{\item[\algorithmicassumption]}
\algrenewcomment[1]{\hfill \(\triangleright\) \emph{\small #1}} %
\algnewcommand{\CommentLine}[1]{\(\triangleright\;\;\) \emph{\small #1} \(\;\;\triangleleft\)} %
\algnewcommand{\InlineIf}[2]{%
\algorithmicif\ #1\ \algorithmicthen\ #2}
\algnewcommand{\InlineIfElse}[3]{%
\algorithmicif\ #1\ \algorithmicthen\ #2\ \algorithmicelse\ #3}
\algnewcommand{\InlineWhile}[2]{%
\algorithmicwhile\ #1\ \algorithmicdo\ #2}
\algnewcommand{\InlineForAll}[2]{%
\algorithmicforall\ #1\ \algorithmicdo\ #2}
\makeatletter
\newcommand{\CWHILE}[2][default]{\algorithmicwhile\ #2\ %
\algorithmicdo\hfill%
\ALC@com{#1}\begin{ALC@whl}}
\newcommand{\CFORALL}[2][default]{\ALC@it\algorithmicforall\ #2\ %
\algorithmicdo\newline%
\ALC@com{#1}\begin{ALC@for}}
\newcommand{\CIF}[2][default]{\ALC@it\algorithmicif\ #2\ %
\algorithmicthen\hfill%
\ALC@com{#1}\begin{ALC@if}}
\newcommand{\FORALLDOEND}[2]{\ALC@it\algorithmicforall\ #1\ \algorithmicdo\ #2 \algorithmicendfor}
\makeatother

\newtheorem{theorem}[algorithm]{Theorem}

\newtheorem{definition}[algorithm]{Definition}
\newtheorem{example}[algorithm]{Example}
\newtheorem{examples}[algorithm]{Examples}

\newtheorem{lemma}[algorithm]{Lemma}

\newtheorem{proposition}[algorithm]{Proposition}
\newtheorem{corollary}[algorithm]{Corollary}
\newtheorem{remark}[algorithm]{Remark}

\newtheorem{problem}{Open question}

\title{Some fast algorithms multiplying a matrix by its adjoint}
\author{Jean-Guillaume Dumas}
\author{Cl\'ement Pernet}
\affil{%
{Univ.\ Grenoble Alpes, umr CNRS 5224 LJK\authorcr F-38000 Grenoble, France}\authorcr
{\small
\url{https://ljk.imag.fr/CAS3C3},
\href{mailto:Jean-Guillaume.Dumas@univ-grenoble-alpes.fr,Clement.Pernet@univ-grenoble-alpes.fr}{\{Jean-Guillaume.Dumas,Clement.Pernet\}@univ-grenoble-alpes.fr}
}}
\author{Alexandre Sedoglavic}
\affil{%
{Univ.\ Lille, CNRS, Centrale Lille, umr 9189 CRIStAL\authorcr F-59000 Lille, France}\authorcr
{\small
\href{http://www.lifl.fr/\%7Esedoglav}{http://www.lifl.fr/\textasciitilde{}sedoglav},
\href{mailto:Alexandre.Sedoglavic@univ-lille.fr}{Alexandre.Sedoglavic@univ-lille.fr}
}}

\begin{document}
\maketitle

\begin{abstract}
We present a non-commutative algorithm for the multiplication of
a~${{2}\times{2}}$ block-matrix by its adjoint, defined by a matrix
ring anti-homo\-morphism.
This algorithm uses~$5$ block products ($3$ recursive calls and~$2$
general products)%
over~$\mathbb{C}$ or in positive characteristic.
The resulting algorithm for arbitrary dimensions is a reduction of
multiplication of a matrix by its adjoint to general matrix product,
improving by a constant factor previously known reductions.
We prove also that there is no algorithm derived from bilinear forms
using only four products and the adjoint of one of them.
Second we give novel dedicated algorithms for the complex field and
the quaternions to alternatively compute the multiplication taking advantage of the structure of the
matrix-polynomial arithmetic involved. We then analyze the respective
ranges of predominance of the two strategies.
Finally we propose schedules with low memory footprint that support a
fast and memory efficient practical implementation over a prime
field.
\end{abstract}

\tableofcontents

\section{Introduction}
Volker Strassen's algorithm~\cite{Strassen:1969:GENO}, with~$7$ recursive
multiplications and~$18$ additions, was the first sub-cubic time
algorithm for matrix product, with a cost of~$\GO{n^{2.81}}$.
Summarizing the many improvements which have happened since then, the
cost of multiplying two arbitrary~$\matrixsize{n}{n}$ matrices over a
ring \Ring~will be
denoted by~${\MM[\Ring]{\omega}(n)=\GO{n^{\omega}}}$ ring operations where $2<\omega\leq 3$ is any
feasible exponent for this operation (see~\cite{LeGall:2014:fmm} for the best theoretical estimates
of~$\omega$ known to date).

We consider here the computation of the
product of a matrix by transpose~${\MatrixProduct{A}{\Transpose{A}}}$ or by its conjugate
transpose~${\MatrixProduct{A}{\CTranspose{A}}}$, which we handle in a unified way as the product
${\MatrixProduct{A}{\MIAH{A}}}$ where~$\SymbolIAH$ is a matrix \emph{anti-homomorphism}.
In the rest of the paper, $\MIAH{A}$ will be referred to as the \emph{adjoint} of~$\mat{A}$.
For this computation, the natural divide and conquer algorithm, splitting the matrices in four
quadrants, will use~$6$ block multiplications (as any of the two off-diagonal blocks can be
recovered from the other one).%
We propose instead a new algorithm
using only~$5$ block multiplications, for any antihomomorphism~$\SymbolIAH$, provided that the base ring supports the existence of skew unitary matrices.

For this product, the best previously known cost bound was
equivalent to~${\frac{2}{2^{\omega}-4}\mathrm{MM}_{\omega}(n)}$ over
any field (see~\cite[\S~6.3.1]{jgd:2008:toms}).
With our algorithm, this product can be computed
in a cost equivalent to~${\frac{2}{2^{\omega}-3}\textrm{MM}_{\omega}(n)}$ ring
operations
when there exists a skew-unitary matrix.
Our algorithm is derived from the class of Strassen-like algorithms
multiplying~$\matrixsize{2}{2}$ matrices in~$7$ multiplications.
Yet it is a reduction of multiplying a matrix by its transpose to
general matrix multiplication, thus supporting any admissible value
for~$\omega$.
By exploiting the symmetry of the problem, it requires about half of
the arithmetic cost of general matrix multiplication when~$\omega$
is~$\log_{2}{7}$.

This paper extends the results of~\cite{jgd:2020:wishart} with the following improvements:
\begin{compactenum}
\item we generalize the case of the transposition in~\cite[Algorithm~2]{jgd:2020:wishart} to arbitrary
  antihomomorphism, including the Hermitian transposition.
\item Our algorithm uses $5$ multiplications and the (hermitian)
  transpose of one these blocks. In~\cite{jgd:2020:wishart} a
  Gr{\"o}bner basis parameterization is used to search for algorithms,
  or prove by exhaustive search that there are no better algorithm,
  \emph{in the Strassen orbit}. We partially address here the more general
  result that there is no algorithm derived from bilinear forms, with fewer products, by proving the
  inexistence of an algorithm with four products and the (hermitian) transpose of one of them.
\item In~\cite{jgd:2020:wishart} the algorithm is shown to be efficient
  over \CC, for a range of matrix multiplication exponents (including
  all the feasible ones), and for any positive characteristic field,
  unconditionally.
  We extend this analysis to the case for the Hermitian transpose: while our five-products algorithm is
  unusable due to the inexistence of skew unitary matrices over~$\CC$, we propose a~2M algorithm,
  adapted from the~3M algorithm for the product of complex matrices.
\item Finally, we propose novel dedicated algorithms for the
  multiplication of a matrix by its transpose or conjugate transpose over the algebra of quaternions
  (over $\RR$ or any commutative field), improving on the dominant term of the state of the art
  complexity bounds for these problems.
\end{compactenum}

After a introducing the terminology in \cref{sec:prelim}, we will present in \cref{sec:algo} the
main recursive algorithm computing the product of a matrix by its adjoint in 5 block products
provided that a skew unitary matrix is given.
We survey in~\cref{sec:skeworthmat} the most classical instances for the base field to support the
existence of skew unitary matrices. We then investigate in~\cref{sec:minimality} the
minimality of five products for the computing the product of a $2\times 2$ matrix by its hermitian
transpose: applying de Groote's technique enables us to state this result partially, for all
algorithms using up to one symmetry between a product and its adjoint.
\cref{sec:fieldext} explores alternative approaches offered by the structure of polynomial
arithmetic, when the field is an extension. This includes a new~2M algorithm in~\cref{ssec:2M} and
new algorithms over the algebra of quaternions in~\cref{ssec:quaternions}. Lastly, we discuss on an
implementation of the recursive algorithm for the product of a matrix by its transpose
in~\cref{sec:implem}.

\section{Preliminaries}
\label{sec:prelim}
To unify the notion of transposition and conjugate transposition, we use the formalism of
antihomomorphisms and of involutive antihomomorphisms as recalled in the following definitions.%
\begin{definition}\label{def:RAH}%
Let~$\Ring,\Sring$ be two rings,~${\SymbolIAH:\Ring\rightarrow\Sring}$ is a \emph{ring
  antihomomorphism} if and only if, for all~${(x,y)}$ in~${\Ring\times\Sring}$:
\begin{subequations}
\begin{align}
& \SymbolIAH(1_\Ring)=1_\Sring\label{RAH:invol},\\
& \SymbolIAH(x+y) = \SymbolIAH(x) + \SymbolIAH(y)\label{RAH:add},\\
& \SymbolIAH(xy) = \SymbolIAH(y)\SymbolIAH(x)\label{RAH:mul}.
\end{align}
\end{subequations}
\end{definition}
From this, one can define a matrix antihomomorphism by induction, as
shown in~\cref{def:MIAH}.
\begin{definition}\label{def:MIAH}%
  Over a ring~$\Ring$, an \emph{involutive matrix antihomomorphism} is a family of applications~${\SymbolIAH_{m,n}:\Matrices{\Ring}{\matrixsize{m}{n}}\rightarrow\Matrices{\Ring}{\matrixsize{n}{m}}}$ for all~${(m,n)}$ in~${\N^{2}}$ satisfying for additional~${(\ell,k)}$ in~${\N^{2}}$ and for all~$\mat{A}$ and~$\mat{A'}$ in~${\Matrices{\Ring}{\matrixsize{m}{n}}}$, for all~$\mat{M}$ in~${\Matrices{\Ring}{\matrixsize{n}{k}}}$, for all~$\mat{B}$ in~${\Matrices{\Ring}{\matrixsize{m}{k}}}$, for all~$\mat{C}$ in~${\Matrices{\Ring}{\matrixsize{\ell}{n}}}$ and for all~$\mat{D}$ in~${\Matrices{\Ring}{\matrixsize{\ell}{k}}}$ the following relations:
\begin{subequations}
\begin{align}
& \SymbolIAH_{m,n}\circ\SymbolIAH_{n,m}=\Identity\label{MMIAH:invol},\\
& \SymbolIAH_{m,n}(\mat{A}+\mat{A'}) = \SymbolIAH_{m,n}(\mat{A}) + \SymbolIAH_{m,n}(\mat{A'})\label{MMIAH:add},\\
& \SymbolIAH_{m,k}(\MatrixProduct{A}{M}) = \SymbolIAH_{n,k}(\mat{M})\cdot\SymbolIAH_{m,n}(\mat{A})\label{MMIAH:mul},\\
& \SymbolIAH_{m+\ell,n+k}\left(\begin{bmatrix}\mat{A}&\mat{B}\\\mat{C}&\mat{D}\end{bmatrix}\right)
= \begin{bmatrix}\SymbolIAH_{m,n}(\mat{A})&\SymbolIAH_{\ell,n}(\mat{C})\\\SymbolIAH_{m,k}(\mat{B})&\SymbolIAH_{\ell,k}(\mat{D})\end{bmatrix}\!.\label{MMIAH:mat}
\end{align}
\end{subequations}
\end{definition}
For the convenience, we will denote all applications of this family by $\phi$, as the dimensions are
clear from the context.
This definition implies the following:

\begin{lemma}\label{lem:transp}
  For all~$\mat{A}$ in~${\Matrices{\Ring}{\matrixsize{m}{n}}}$ let~$\mat{B}$, be~$\MIAH{\mat{A}}$. Then for all suitable~$(i,j)$ the coefficient~$b_{ij}$ is~$\SymbolIAH(a_{ji})$.
\end{lemma}
\begin{proof}
  By induction, using~\cref{MMIAH:mat}: if~${m=n=1}$, then~${\MIAH{{A}}=[\SymbolIAH(a_{11})]}$. Then
  assume the property is true for all $\mat{A}\in\Ring^{\matrixsize{m}{n}}$ with $m,n\leq N$, and
  consider a matrix~$\mat{A}$ in~$\Matrices{\Ring}{\matrixsize{(N+1)}{(N+1)}}$. Applying  \cref{MMIAH:mat} on the
  block decomposition $\mat{A}=
  \begin{bmatrix}
    \mat{A}_{11} & \mat{a}_{12}\\
    \mat{a}_{21} & a_{22}\\
  \end{bmatrix}$
  where $\mat{A}_{11}$ is in~${\Matrices{\Ring}{\matrixsize{N}{N}}}$ yields the relations:
	\begin{equation}
  \MIAH{\mat{a}} =  \begin{bmatrix}
    \MIAH{\mat{A}_{11}} & \MIAH{\mat{a}_{21}}\\
      \MIAH{\mat{a}_{12}} & \SymbolIAH(a_{22})\\
  \end{bmatrix} = [\SymbolIAH(a_{ji})]_{ij}
\end{equation} by  induction hypothesis.
  The case of matrices in $\Ring^{\matrixsize{m}{(N+1)}}$ and~$\Matrices{\Ring}{{\matrixsize{(N+1)}{n}}}$ is dealt
  with similarly, using 0-dimensional blocks $\mat{a}_{21}$ or $\mat{a}_{12}$ respectively.
  \end{proof}

\begin{lemma}\label{lem:constant}
  For all~$\alpha$ in~${\Ring}$ and for all~$\mat{A}$ in~${\Matrices{\Ring}{\matrixsize{m}{n}}}$,~${\MIAH{\alpha{\mat{A}}} = \MIAH{{A}}\MIAH{\alpha}}$.
\end{lemma}
\begin{proof}
  By~\cref{lem:transp},~${\MIAH{\alpha \IdentityMatrixOfSize{m}} =\SymbolIAH(\alpha)\IdentityMatrixOfSize{m}= \IdentityMatrixOfSize{m}\SymbolIAH(\alpha)}$.
  Then by~\cref{MMIAH:mul}, the relations~${\MIAH{\alpha{\mat{A}}}=\MIAH{(\alpha \IdentityMatrixOfSize{m})  \mat{A} } = \MIAH{{A}} \MIAH{\alpha \IdentityMatrixOfSize{m}} = \MIAH{{A}} \SymbolIAH(\alpha)}$ hold.
\end{proof}

The following~\cref{lem:endo} shows that~\cref{def:MIAH} is a natural
extension of a ring antiendomorphism for matrices.

\begin{lemma}\label{lem:endo} 
An involutive matrix antihomomorphism is a ring antiendomorphism on its base ring (seen as the ring of $1{\times}1$ matrices).
\end{lemma}
\begin{proof} \cref{MMIAH:add,MMIAH:mul} directly
  imply~\cref{RAH:add,RAH:mul} respectively when $m=n=k=1$.
 Then, we have that $\SymbolIAH(1)=\SymbolIAH(1)\cdot{1}$.
 Therefore $\SymbolIAH(\SymbolIAH(1))=\SymbolIAH(\SymbolIAH(1)\cdot{1})$ and
 $1=\SymbolIAH(1)\SymbolIAH(\SymbolIAH(1))=\SymbolIAH(1)\cdot{1}$
 by~\cref{MMIAH:invol,MMIAH:mul,MMIAH:invol}. This right hand side is
 equal to that of the first equation, thus proving the equality of the
 left hand sides and \cref{RAH:invol}.
\end{proof}

\cref{def:MIAH} gathers actually all the requirements for our
algorithm to work in classical hermitian or non-hermitian cases:
\begin{examples}\label{ex:antih}
  For matrices over a commutative ring,
  \begin{itemize}
  \item the matrix transpose with $\MIAH{A}=\Transpose{A}$ and
  \item the matrix conjugate transpose, $\MIAH{A}=\HermitianTranspose{A}$,
  \end{itemize}
  are two examples of matrix anti-homomorphisms.
  However, for instance, transposition over the quaternions is a
  counter-example as the non-commutativity implies there that in
  general~${\Transpose{\left(\MatrixProduct{A}{B}\right)}\neq\MatrixProduct{\Transpose{B}}{\Transpose{A}}}$.
\end{examples}
\begin{definition}
The image $\MIAH{A}$ of a matrix~$\mat{A}$ by an antihomomorphism is called the \emph{adjoint} of~$\mat{A}$.
\end{definition}

\begin{definition}
  Let $\mat{A}\in\Ring^{m{\times}n}$, we denote respectively by \Low{A} and \Up{A} the $m\times n$ lower
and upper triangular parts of $\mat{A}$, namely the matrices $\mat{L}$ and $\mat{U}$ verifying
\begin{itemize}
\item $\mat{L}_{ij}=a_{ij}$ for $i\geq j$ and $\mat{L}_{ij}=0$ otherwise,
\item $\mat{U}_{ij}=a_{ij}$ for $i\leq j$ and $\mat{U}_{ij}=0$ otherwise.
\end{itemize}
\end{definition}

\begin{lemma}\label{lem:uplow}
If~${\MIAH{\mat{A}}=\mat{A}}$ in~${\Matrices{\Ring}{\matrixsize{n}{n}}}$, then $\Up{{A}}=\MIAH{\Low{{A}}}$.
\end{lemma}
\begin{proof}
  Applying~\cref{lem:transp}, the coefficients $u_{ij}$ of $\mat{U}=\MIAH{\Low{{A}}}$ for $0<i\leq j$
  satisfy $u_{ij}= \SymbolIAH(a_{ji})$. Now if $\MIAH{\mat{A}}=\mat{A}$, we have  $u_{ij}= a_{ij}$
  for $0<i\leq j$ and $u_{ij}=0$ otherwise, as $\MIAH{0}=0$, by
  \cref{MMIAH:add}. Hence $\mat{U}=\Up{{A}}$.
\end{proof}
\begin{definition}[Skew-unitary]
A matrix~$\mat{Y}$ in~${\Ring^\matrixsize{n}{n}}$ is \emph{skew-unitary} relatively to a matrix antihomomorphism $\SymbolIAH$ if the following relation holds:
\begin{equation}\label{eq:skewu}
\MatrixProduct{\mat{Y}}{\MIAH{{Y}}}=-\IdentityMatrixOfSize{n}.
\end{equation}
\end{definition}

For the cost analysis, we will also need the following variant of the Master Theorem, reflecting the
constant in the leading term of the computed cost bound.
\begin{lemma}
  \label{lem:masterthm}
Let $T(n)$ be defined by the recurrence $T(n)=aT(n/2)+b\left(\frac{n}{2}\right)^\alpha + \LO{n^\alpha}$,
where $0\leq \log_2 a < \alpha$.
Then $T(n) = \frac{b}{2^\alpha-a}n^\alpha + \LO{n^\alpha}$.
\end{lemma}
\begin{proof}
  \begin{equation*}
    \begin{split}
      T(n) &= a^{\log_2 n}T(1)+ \sum_{i=0}^{\log_2(n)-1} a^i b\left(\frac{n}{2^{i+1}}\right)^\alpha+ \LO{\left(\frac{n}{2^{i}}\right)^\alpha}\\
   & = n^{\log_2 a}T(1) + \frac{b}{2^\alpha}n^\alpha \sum_{i=0}^{\log_2(n)-1}\left(\frac{a}{2^\alpha}\right)^i%
    + \LO{n^\alpha}= \frac{b}{2^\alpha-a} n^\alpha + \LO{n^\alpha}.
    \end{split}
  \end{equation*}
\end{proof}
\section{An algorithm for the product of a matrix by its adjoint with five multiplications}
\label{sec:algo}
We now show how to compute the product of a matrix by its adjoint with respect to
an involutive antihomomorphism in only $5$ recursive multiplications
and $2$ multiplications by any skew-unitary matrix. This is a
generalization of~\cite[Algorithm~2]{jgd:2020:wishart} for any involutive
antihomomorphism.

We next give~\cref{alg:MIAHMM} for even dimensions. In case of odd
dimensions, padding or static/dynamic peeling can always be
used~\cite{Dumas:2009:WinoSchedule}.

\begin{algorithm}[htb]
\caption{Product of a matrix by its adjoint}\label{alg:MIAHMM}
\begin{algorithmic}
\Require{$\mat{A}\in \Ring^{m\times n}$ (with even $m$ and $n$ for the
  sake of simplicity);}
\Require{$\SymbolIAH$ an involutive matrix antihomomorphism;}
\Require{$\mat{Y} \in \Ring^{\frac{n}{2}\times\frac{n}{2}}$
  skew-unitary for $\SymbolIAH$.}
\Ensure{$\Low{\MatrixProduct{\mat{A}}{\MIAH{{A}}}}$.}
	\State Split $\mat{A}=\begin{smatrix} \mat{A}_{11}&\mat{A}_{12}\\ \mat{A}_{21}&\mat{A}_{22}\end{smatrix}$ where $\mat{A}_{11}$ is in~${\Matrices{\Ring}{\matrixsize{\frac{m}{2}}{\frac{n}{2}}}}$
\State \Comment{$4$ additions and 2 multiplications by~$\mat{Y}$:}
\State \({\mat{S}_{1}}\leftarrow{\MatrixProduct{(\mat{A}_{21} - \mat{A}_{11})}{\mat{Y}}}\)
\State \({\mat{S}_{2}}\leftarrow{\mat{A}_{22} - \MatrixProduct{\mat{A}_{21}}{\mat{Y}}}\)
\State \({\mat{S}_{3}}\leftarrow{\mat{S}_{1} - \mat{A}_{22}}\)
\State \({\mat{S}_{4}}\leftarrow{\mat{S}_{3} + \mat{A}_{12}}\)
\State \Comment{$3$ recursive (${\mat{P}_{1}, \mat{P}_{2}, \mat{P}_{5}}$) and~$2$ general products (${\mat{P}_{3}, \mat{P}_{4}}$):}
\State
\({\Low{{P}_{1}}\leftarrow\Low{\MatrixProduct{\mat{A}_{11}}{\MIAH{{A}_{11}}}}}\)
\State \({\Low{{P}_{2}}\leftarrow\Low{\MatrixProduct{\mat{A}_{12} }{\MIAH{{A}_{12}}}}}\)
\State \({\mat{P}_{3} \leftarrow\MatrixProduct{\mat{A}_{22}}{\MIAH{{S}_{4}}}}\)
\State \({\mat{P}_{4} \leftarrow\MatrixProduct{\mat{S}_{1}}{\MIAH{{S}_{2}}}}\)
\State \( {\Low{{P}_{5}}\leftarrow\Low{\MatrixProduct{\mat{S}_{3}}{\MIAH{{S}_{3}}}}} \)
\State \Comment{$3$ half additions and~$2$ complete additions:}
\State \(\Low{{U}_{1}}\!\leftarrow\!\Low{{P}_{1}}\!+\!\Low{{P}_{5}}\)
\State \(\Low{{U}_{3}}\!\leftarrow\!\Low{{P}_{1}}\!+\!\Low{{P}_{2}}\)
\State \(\text{Up}(\mat{U}_{1})\leftarrow\MIAH{\Low{{U}_{1}}}\) \Comment{Forms
  the full matrix $\mat{U}_1$}
\State \( \mat{U}_{2} \leftarrow \mat{U}_{1} + \mat{P}_{4} \),
\State \( \mat{U}_{4} \leftarrow \mat{U}_{2} + \mat{P}_{3} \),
\State \(\Low{{U}_{5}} \leftarrow \Low{{U}_{2}} + \Low{\MIAH{{P}_{4}}}.\)
\State \Return{$\begin{smatrix} \Low{{U}_{3}} &\\ \mat{U}_{4} & \Low{{U}_{5}} \end{smatrix}$.}
\end{algorithmic}
\end{algorithm}

\begin{theorem}\label{thm:complexitybound}
  \cref{alg:MIAHMM} is correct.
Moreover, if any two $n\times{n}$ matrices over a ring $\Ring$ can be multiplied
in  $\MM[\Ring]{\omega}(n)=\GO{n^{\omega}}$ ring
operations for $\omega>2$,
and if there exist a skew-unitary matrix which can be multiplied to
any other matrix in $\LO{n^{\omega}}$ ring operations
then
\cref{alg:MIAHMM} requires fewer than~${\frac{2}{2^{\omega}-3}\MM[\Ring]{\omega}(n)+\LO{n^{\omega}}}$
ring operations.

\end{theorem}

\begin{proof}
For the cost analysis,  \cref{alg:MIAHMM} is applied recursively to compute three
products~$P_{1}, P_{2}$ and~$P_{7}$, while~$P_{4}$ and~$P_{5}$ are
computed
in~$\MM[\Ring]{\omega}(n)$ using the general matrix multiplication algorithm.
The second hypothesis is that applying the skew-unitary matrix~$Y$ to a~${{n}\times{n}}$ matrix costs~$Y(n)=\LO{n^\omega}$.
Then applying \cref{rq:7.5} thereafter, the cost~$T(n)$ of \cref{alg:MIAHMM} satisfies:
\begin{equation}\label{eq:complexity}
T(n) \leq 3T(n/2) + 2\MM[\Ring]{\omega}(n/2) + 2Y(n) + (7.5){(n/2)}^{2} +
\LO{n^2}
\end{equation}
and~$T(4)$ is a constant.
Thus, by~\cref{lem:masterthm}:
\begin{equation}
	T(n) \leq \frac{2}{2^{\omega}-3}\MM[\Ring]{\omega}(n) +\LO{n^{\omega}}.
\end{equation}

Now for the correction, by~\cref{MMIAH:mat}, we have to show that the result of~\cref{alg:MIAHMM} is indeed:
\[\begin{split}
\Low{\MatrixProduct{\mat{A}}{\MIAH{{A}}}} & =
\Low{\MatrixProduct{\mat{A}}{\begin{smatrix}\MIAH{{A}_{11}}&\MIAH{{A}_{21}}\\\MIAH{{A}_{12}}&\MIAH{{A}_{22}}\end{smatrix}}}\\
&=\begin{smatrix}
  \Low{\MatrixProduct{\mat{A}_{11}}{\MIAH{{A}_{11}}}+\MatrixProduct{\mat{A}_{12}}{\MIAH{{A}_{12}}}} & \times \\
\MatrixProduct{\mat{A}_{21}}{\MIAH{{A}_{11}}}+\MatrixProduct{\mat{A}_{22}}{\MIAH{{A}_{12}}}
&\Low{\MatrixProduct{\mat{A}_{21}}{\MIAH{{A}_{21}}}+\MatrixProduct{\mat{A}_{22}}{\MIAH{{A}_{22}}}}
\end{smatrix}
\end{split}\]

First, we have that:
\begin{equation}
\label{eq:prf:u3}
	\Low{{U}_{3}} = \Low{{P}_{1}}+\Low{{P}_{2}}
        = \Low{\MatrixProduct{\mat{A}_{11}}{\MIAH{{A}_{11}}}
	+\MatrixProduct{\mat{A}_{12}}{\MIAH{{A}_{12}}}}.
\end{equation}

Second, as~$\mat{Y}$ is skew-unitary, then we have that~${\MatrixProduct{\mat{Y}}{\MIAH{{Y}}}=-\IdentityMatrixOfSize{\frac{n}{2}}}$.
Also, by~\cref{MMIAH:add,MMIAH:mul},
$\MIAH{{S}_{2}}=\MIAH{{A}_{22}}-\MatrixProduct{\MIAH{{Y}}}{\MIAH{{A}_{21}}}$.
Then, denote by~$\mat{R}_{1}$ the product:
\begin{equation}
\label{eq:prf:r1}
\begin{split}
\mat{R}_{1} &= \MatrixProduct{\mat{A}_{11}}{\MatrixProduct{\mat{Y}}{\MIAH{{S}_{2}}}}
= \MatrixProduct{\mat{A}_{11}}{\MatrixProduct{\mat{Y}}{(\MIAH{{A}_{22}} -
\MatrixProduct{\MIAH{{Y}}}{\MIAH{{A}_{21}}})}}\\
& = \MatrixProduct{\mat{A}_{11}}{(\MatrixProduct{\mat{Y}}{\MIAH{{A}_{22}}} + \MIAH{{A}_{21}})}.
\end{split}
\end{equation}

Further, by~\cref{MMIAH:invol,MMIAH:mul}, we have that
$\mat{P}_{1}=\MatrixProduct{\mat{A}_{11}}{\MIAH{A_{11}}}$, $\mat{P}_{2}=\MatrixProduct{\mat{A}_{12}}{\MIAH{{A}_{12}}}$,
and~${\mat{P}_{5}=\MatrixProduct{\mat{S}_{3}}{\MIAH{{S}_{3}}}}$ are invariant under
the action of \SymbolIAH.
So are therefore, ${\mat{U}_{1}=\mat{P}_{1}+\mat{P}_{5}}$, ${\mat{U}_{3}=\mat{P}_{1}+\mat{P}_{2}}$
and~${\mat{U}_{5}=\mat{U}_{1}+(\mat{P}_{4}+\MIAH{{P}_{4}})}$.
By \cref{lem:uplow}, it suffices to compute $\Low{{U}_{1}}$ and, if needed,
we also have $\Up{{U}_{1}}=\MIAH{\Low{{U}_{1}}}$.

Then, as~$\mat{S}_{3} = \mat{S}_{1}-\mat{A}_{22} =
  \MatrixProduct{(\mat{A}_{21}-\mat{A}_{11})}{\mat{Y}}-\mat{A}_{22}=-\mat{S}_{2}-\MatrixProduct{\mat{A}_{11}}{\mat{Y}}$,
and $\MIAH{\MIAH{{S}_2}}=\mat{S}_2$ by~\cref{MMIAH:invol}, we have that:
\begin{equation}
\label{eq:prf:u1}
\begin{split}
	\mat{U}_{1} & = \mat{P}_{1} + \mat{P}_{5}  = \MatrixProduct{\mat{A}_{11}}{\MIAH{{A}_{11}}} + \MatrixProduct{\mat{S}_{3}}{\MIAH{{S}_{3}}} \\
	& = \MatrixProduct{\mat{A}_{11}}{\MIAH{{A}_{11}}}
	+ \MatrixProduct{(\mat{S}_{2}+\MatrixProduct{\mat{A}_{11}}{\mat{Y}})}{(\MIAH{{S}_{2}}
	+\MatrixProduct{\MIAH{{Y}}}{\MIAH{{A}_{11}}} )} \\
	& = \MatrixProduct{\mat{S}_{2}}{\MIAH{{S}_{2}}}+\MIAH{{R}_{1}}+\mat{R}_{1}.
\end{split}
\end{equation}

Also, denote~${\mat{R}_{2}=\MatrixProduct{\mat{A}_{21}}{\MatrixProduct{{\mat{Y}}}{\MIAH{{A}_{22}}}}}$, so that:
\begin{equation}
\label{eq:prf:s2s2T}
\begin{split}
	\MatrixProduct{\mat{S}_{2}}{\MIAH{{S}_{2}}} & = \MatrixProduct{(\mat{A}_{22} -
	\MatrixProduct{\mat{A}_{21}}{\mat{Y}})}{(\MIAH{{A}_{22}} -
	\MatrixProduct{\MIAH{{Y}}}{\MIAH{{A}_{21}}})} \\
	& = \MatrixProduct{\mat{A}_{22}}{\MIAH{{A}_{22}}}
		-\MatrixProduct{\mat{A}_{21}}{\MIAH{{A}_{21}}}-\mat{R}_{2}-\MIAH{{R}_{2}}.
\end{split}
\end{equation}

Furthermore, from~\cref{eq:prf:r1}:
\begin{equation}
\label{eq:prf:r1p4}
\begin{split}
	  \mat{R}_{1} + \mat{P}_{4} & = \mat{R}_{1} + \MatrixProduct{\mat{S}_{1}}{\MIAH{{S}_{2}}}\\
	& = \mat{R}_{1} + \MatrixProduct{(\mat{A}_{21} - \mat{A}_{11})}{\MatrixProduct{\mat{Y}}{(\MIAH{{A}_{22}} - \MatrixProduct{\MIAH{{Y}}}{\MIAH{{A}_{21}}})}}\\
	& = \MatrixProduct{\mat{A}_{21}}{\MatrixProduct{\mat{Y}}{\MIAH{{A}_{22}}}}
	+ \MatrixProduct{\mat{A}_{21}}{\MIAH{{A}_{21}}}
	  = \mat{R}_{2} + \MatrixProduct{\mat{A}_{21}}{\MIAH{{A}_{21}}}.
\end{split}
\end{equation}

This shows, from~\cref{eq:prf:u1,eq:prf:s2s2T,eq:prf:r1p4}, that:
  \begin{equation}\label{eq:prf:u5}\begin{split}
      \mat{U}_{5} & = \mat{U}_{1} + \mat{P}_{4} + \MIAH{{P}_{4}}
	   = \MatrixProduct{\mat{S}_{2}}{\MIAH{{S}_{2}}} + \MIAH{{R}_{1}}+ \mat{R}_{1}+\mat{P}_{4}+\MIAH{{P}_{4}} \\
	  & = \MatrixProduct{\mat{A}_{22}}{\MIAH{{A}_{22}}}
	  +(-1+2)\MatrixProduct{\mat{A}_{21}}{\MIAH{{A}_{21}}}.
    \end{split}\end{equation}

Third, the last coefficient~$\mat{U}_{4}$ of the result is obtained
from~\cref{eq:prf:r1p4,eq:prf:u5}:

  \begin{equation}\label{eq:prf2:u4}\begin{split}
	  &  \mat{U}_{4} = \mat{U}_{2} + \mat{P}_{3} = \mat{U}_{1} + \mat{P}_{4} +\mat{P_3}\\
	  & = \MatrixProduct{\mat{A}_{22}}{\MIAH{{A}_{22}}}
      -\MatrixProduct{\mat{A}_{21}}{\MIAH{{A}_{21}}}-\mat{R}_{2}-\MIAH{{R}_{2}}+\mat{R}_{1}+\MIAH{{R}_{1}}
      +\mat{P}_4 + \mat{P}_3\\
	  & = \MatrixProduct{\mat{A}_{22}}{\MIAH{{A}_{22}}}
      -\MIAH{{R}_{2}}+\MIAH{{R}_{1}}
      + \mat{P}_3\\
       \end{split}
\end{equation}
  since by~\cref{eq:prf:r1p4},
  $\mat{R}_{1}+\mat{P}_4 =\mat{R_2}+\MatrixProduct{\mat{A}_{21}}{\MIAH{{A}_{21}}}$.
  Now
  \begin{equation}
    \begin{split}
     & \mat{P}_3 =  \MatrixProduct{\mat{A}_{22}}{\MIAH{{S}_4}} =
      \MatrixProduct{\mat{A}_{22}}{\MIAH{\MatrixProduct{(\mat{A}_{21}-\mat{A}_{11})}{\mat{Y} +
            \mat{A}_{12} -\mat{A}_{22}}}}\\
      & = \MIAH{{R}_2} - \MIAH{{R}_1} +\MatrixProduct{\mat{A}_{21}}{\MIAH{{A}_{11}}},
    \end{split}
  \end{equation}
  Hence
  $$
 \mat{U}_{4} =  \MatrixProduct{\mat{A}_{22}}{\MIAH{{A}_{12}}} +\MatrixProduct{\mat{A}_{21}}{\MIAH{{A}_{11}}}
  $$
\end{proof}

To our knowledge, the best previously known result was with a~$\frac{2}{2^{\omega}-4}$ factor instead, see e.g.~\cite[\S~6.3.1]{jgd:2008:toms}.
Table~\ref{tab:ffcomplex} summarizes the arithmetic complexity bound improvements.
\begin{table}[htbp]\centering%
\begin{tabular}{lcrrr}
\toprule
Problem & Alg.\ & $\GO{n^{3}}$ & $\GO{n^{\log_2(7)}}$ & $\GO{n^{\omega}}$ \\
\midrule
\multirow{2}{*}{$\MatrixProduct{A}{\MIAH{A}} \in\F^{n{\times}n}$}
& \cite{jgd:2008:toms} & $n^{3}$ & $\frac{2}{3}\,\textrm{MM}_{\log_{2}(7)}(n)$& $\frac{2}{2^{\omega}-4}\,\textrm{MM}_{\omega}(n)$\\
& Alg.~\ref{alg:MIAHMM}  & \color{darkred}\bf\boldmath$0.8 n^{3}$ & \color{darkred}\bf\boldmath$\frac{1}{2}\,\textrm{MM}_{\log_{2}(7)}(n)$& \color{darkred}\bf\boldmath$\frac{2}{2^{\omega}-3}\,\textrm{MM}_{\omega}(n)$ \\
\bottomrule
\end{tabular}
\caption{Arithmetic complexity bounds leading terms.}\label{tab:ffcomplex}
\end{table}

\begin{examples} In many cases, applying the skew-unitary
  matrix~$Y$ to a~${{n}\times{n}}$ matrix costs only~$yn^{2}$ for some
  constant~$y$ depending on the base ring.
  If the ring is the complex field~$\mathbb{C}$ or satisfies the conditions of
  \cref{lem:lemsqrt}, there is a square root~$i$
  of~$-1$. Setting~${Y=i\,\IdentityMatrixOfSize{n/2}}$
  yields~${Y(n)=n^2}$.
  Otherwise, we show in~\cref{sec:skeworthmat} that in
  characteristic~${p\equiv{3}\bmod{4}}$, \cref{lem:pigeonhole}
  produces~$Y$ equal
  to~${\begin{smatrix}a&b\\-b&a\end{smatrix}\tensorproduct{\IdentityMatrixOfSize{n/2}}}$
  for which~${Y(n)=3n^2}$.
  As a sub-case, the latter can be improved when~${p\equiv{3}\bmod{8}}$:
  then, \cref{lem:minustwo} shows that~$-2$ is a square.
  Therefore, in this case set~${a=1}$ and~${b\equiv\sqrt{-2}\bmod{p}}$
  such that one multiplication is saved. Then the relation~${a^{2}+b^{2}=-1}$
  there yields~${Y=\begin{smatrix} 1 & \sqrt{-2}\\ -\sqrt{-2} &
      1\end{smatrix}\tensorproduct{\IdentityMatrixOfSize{n/2}}}$ for
  which~${Y(n)=2n^2}$.
\end{examples}

\begin{remark}\label{rq:7.5}
Each recursive level of \cref{alg:MIAHMM} is composed of 9 block additions.
An exhaustive search on all symmetric algorithms in the orbit of that
of Strassen (via a Gr\"obner basis
parameterization~\cite{jgd:2020:wishart}) showed that this number is
minimal in this class of algorithms.
Note also that~$3$ out of these~$9$ additions in \cref{alg:MIAHMM}
involve symmetric matrices and are therefore only performed on the
lower triangular part of the matrix.
Overall, the number of scalar additions
is~${6n^{2}+3/2n(n+1)=15/2n^{2}+1.5n}$, nearly half of the optimal in
the non-symmetric case~\cite[Theorem~1]{bshouty:1995a}.
\end{remark}

To further reduce the number of additions, a promising approach is
that undertaken
in~\cite{Karstadt:2017:strassen,Beniamini:2019:fmmsd}.
This is however not clear to us how to adapt our strategy to their
recursive transformation of basis.

\section{Rings with skew unitary matrices}
\label{sec:skeworthmat}
\cref{alg:MIAHMM} requires a skew-unitary matrix.
Unfortunately there are no skew-unitary matrices over~$\RR$, nor~$\QQ$
for $\SymbolIAH$ the transposition, nor over $\CC$ for $\SymbolIAH$ the
Hermitian transposition (there $-1$ cannot be a sum of real squares for a
diagonal element of $\MatrixProduct{\mat{Y}}{\MIAH{Y}}$).
Hence, \cref{alg:MIAHMM} provides no improvement in these cases.
In other domains, the simplest skew-unitary matrices just use a
square root of~$-1$ while others require a sum of squares.

\subsection{Over the complex field}
\cref{alg:MIAHMM} is thus directly usable over~${\CC^{\matrixsize{n}{n}}}$ with~${\MIAH{A}=\Transpose{\mat{A}}}$ and~${\mat{Y}=i\,\IdentityMatrixOfSize{\frac{n}{2}}}$ in~$\CC^{\matrixsize{\frac{n}{2}}{\frac{n}{2}}}$.
When complex numbers are represented in Cartesian form, as a pair of real numbers,
the multiplications by~${\mat{Y}=i\,\IdentityMatrixOfSize{\frac{n}{2}}}$ are
essentially free since they just exchange the real and imaginary
parts, with one sign flip.

As mentioned, for the conjugate transposition,~${\MIAH{A}=\HermitianTranspose{\mat{A}}}$, on the contrary, there are no candidate skew-unitary matrices and we for now report no
improvement in this case using this approach (but another one does as shown in~\cref{ssec:2M}).

Now, even though over the complex the product of a
matrix by its \emph{conjugate} transpose is more widely used,
there are some applications for the product of a matrix by its transpose, see for
instance~\cite{Baboulin:2005:csyrk}. This is reflected in the \textsc{blas} \textsc{api}, where both routines
\texttt{zherk} and \texttt{zsyrk} are offered.

\subsection[Rings where negative one is a square and transposition]{Rings where negative one is a square and $\MIAH{A}=\Transpose{\mat{A}}$}
Over some rings , square roots of~$-1$ can also
be elements of the base field, denoted~$i$ in~$\Ring$ again.
There, \cref{alg:MIAHMM} only requires some
pre-multiplications by this square root (with
also~${\mat{Y}=i\,\IdentityMatrixOfSize{\frac{n}{2}}\in\Ring^{\matrixsize{\frac{n}{2}}{\frac{n}{2}}}}$), but
\emph{within the ring}.

Further, when the ring is a field in positive characteristic, the existence of a square root of
minus one can be characterized, as shown in~\cref{lem:lemsqrt}, thereafter.
\begin{proposition}\label{lem:lemsqrt}
Fields with characteristic two, $p$ satisfying~${{p}\equiv{1}\bmod{4}}$, or finite fields that are an even extension, contain a square root of~$-1$.
\end{proposition}
\begin{proof}
  If~${p=2}$, then~${1=1^{2}=-1}$.
  If~${{{p}\equiv{1}}\bmod{4}}$, then half of the non-zero elements~$x$ in the base field of size~$p$ satisfy~${x^{\frac{p-1}{4}} \neq \pm 1}$ and then the square of the latter must be~$-1$.
  If the finite field~$\F$ is of cardinality~$p^{2k}$, then, similarly, there exists elements~${x^{\frac{p^{k}-1}{2}\frac{p^{k}+1}{2}}}$ different from~$\pm 1$ and then the square of the latter must be~$-1$.
\end{proof}

\subsection[Any field with positive characteristic and transposition]{Any field with positive characteristic and $\MIAH{A}=\Transpose{A}$}\label{ssec:tridiag}
Actually, we show that \cref{alg:MIAHMM} can also be run without any field extension, even when~$-1$ is not a square:
form the skew-unitary matrices constructed in
\cref{lem:pigeonhole}, thereafter, and use them directly as
long as the dimension of~$\mat{Y}$ is even.
Whenever this dimension is odd, it is always possible to pad with zeroes so that~${\MatrixProduct{\mat{A}}{\Transpose{A}}=\MatrixProduct{\begin{smatrix}\mat{A}&0\end{smatrix}}{\begin{smatrix} \Transpose{A} \\ 0\end{smatrix}}}$.
\begin{proposition}\label{lem:pigeonhole}
  Let~$\F_{p^k}$ be a field of characteristic~$p$, there exists~${(a,b)}$ in~${\F_p^{2}}$
  such that the matrix:
\begin{equation}
\begin{smatrix}
a & b\\
-b & a
\end{smatrix}\tensorproduct{\IdentityMatrixOfSize{n}} =
\begin{smatrix}
a\, \IdentityMatrixOfSize{n} & b\, \IdentityMatrixOfSize{n}\\
-b\, \IdentityMatrixOfSize{n} & a\, \IdentityMatrixOfSize{n}
\end{smatrix}\quad \textrm{in}\quad \F_p^{2n{\times}2n}
\end{equation}
is skew-unitary for the transposition.
\end{proposition}

\begin{proof}
Using the relation
\begin{equation}
\begin{smatrix}
a \,\IdentityMatrixOfSize{n} & b \,\IdentityMatrixOfSize{n}\\
-b \,\IdentityMatrixOfSize{n} & a \,\IdentityMatrixOfSize{n}
\end{smatrix}
\Transpose{%
\begin{smatrix}
a \,\IdentityMatrixOfSize{n} & b \,\IdentityMatrixOfSize{n}\\
-b \,\IdentityMatrixOfSize{n} & a \,\IdentityMatrixOfSize{n}
\end{smatrix}}=
(a^2+b^2)\,\IdentityMatrixOfSize{2n},
\end{equation}
it suffices to prove that there exist~$a,b$ such that~${a^{2}+b^{2}=-1}$.
In characteristic~2,~${{a=1},{b=0}}$ is a solution as~${1^{2}+0^{2}=-1}$.
In odd characteristic, there are~${\frac{p+1}{2}}$ distinct square elements~${x_{i}}^{2}$ in the base prime field.
Therefore, there are~$\frac{p+1}{2}$ distinct elements~${-1-{x_{i}}^{2}}$.
But there are only~$p$ distinct elements in the base field, thus there exists a couple~$(i,j)$ such
that~${-1-{x_{i}}^{2}}={x_{j}}^{2}$~\cite[Lemma~6]{Seroussi:1980:BBgfp}.
\end{proof}

To further improve the running time of multiplications by a
skew-unitary matrix in this case, one could set one of the squares to
be $1$. This is possible if $-2$ is a square, for instance when
${p\equiv{3}\bmod{8}}$:
\begin{lemma}\label{lem:minustwo}
If ${p\equiv{3}\bmod{8}}$ then $-2$ is a square modulo $p$.
\end{lemma}
\begin{proof} Using Legendre symbol,
$\left(\frac{-2}{p}\right)=\left(\frac{-1}{p}\right)\left(\frac{2}{p}\right)=(-1)^{\frac{p-1}{2}}(-1)^{\frac{p^2-1}{8}}=(-1)(-1)=1$
\end{proof}

Now, \cref{lem:pigeonhole} shows that skew-unitary matrices
do exist for any field with positive characteristic.
For \cref{alg:MIAHMM}, we need to build them mostly for~${{p}\equiv{3}\bmod 4}$ (otherwise use \cref{lem:lemsqrt}).
\par
For this, without the extended Riemann hypothesis (\mathsc{erh}), it is possible to use the decomposition of primes into squares:
\begin{enumerate}
\item Compute by enumeration a prime~${r=4pk+(3-1)p-1}$, so that
  both relations~$r\equiv{1}\bmod{4}$ and~$r\equiv{-1}\bmod{p}$ hold;
\item Thus, the methods of~\cite{brillhart:1972:twosquares} allow one
  to decompose any prime into squares and give a couple~${(a,b)}$
  in~${\Z^{2}}$ such that~${a^2+b^2=r}$.
Finally, this gives~$a^{2}+b^{2}\equiv{-1}\bmod{p}$.
\end{enumerate}
By the prime number theorem the first step is polynomial in~$\log(p)$,
as is the second step (square root modulo a prime, denoted \FFSqrtname, has a cost close to exponentiation and then the rest of Brillhart's algorithm is \mathsc{gcd}-like).
In practice, though, it is faster to use the following
\cref{alg:sosmodp}, even though the latter has a better
asymptotic complexity bound only if the \mathsc{erh} is true.
\begin{algorithm}[htbp]\caption{\texttt{SoS}: Sum of squares decomposition over a finite field}\label{alg:sosmodp}
\begin{algorithmic}[1]
 \Require{${p\in\Primes\backslash\{2\}}$,~${k\in\Z}$.}
 \Ensure{${(a,b)\in\Z^{2}}$, s.t.~${a^{2}+b^{2}\equiv{k}\bmod{p}}$.}
 \If{$\left(\frac{k}{p}\right)=1$}
    \Comment{$k$ is a square mod~$p$}
    \State \Return{$\left(\FFSqrt{p}{k},0\right)$.}
 \Else \Comment{Find smallest quadratic non-residue}
 \State $s\assign 2$;
 \InlineWhile{$\left(\frac{s}{p}\right)==1$}{$s\assign s+1$}
 \EndIf
 \State ${c \assign \FFSqrt{p}{s-1}}$ \hfill\Comment{${s-1}$ must be a square}
 \State $r \assign k s^{-1} \bmod{p}$
 \State ${a \assign \FFSqrt{p}{r}}$ \Comment{Now~${{k}\equiv{a^{2}s}\equiv{a^{2}(1+c^{2})}\bmod{p}}$}
 \State \Return $\left(a, ac\bmod{p}\right)$
\end{algorithmic}
\end{algorithm}

\begin{proposition}\label{thm:sosmodpcorrect}
\cref{alg:sosmodp} is correct and, under the~\mathsc{erh}, runs in
expected time~${\widetilde{O}\bigl({\log}^{3}(p)\bigr)}$.
\end{proposition}
\begin{proof}
If~$k$ is square then the square of one of its square roots added to
the square of zero is a solution.
Otherwise, the lowest quadratic non-residue (\mathsc{lqnr}) modulo~$p$
is one plus a square~$b^{2}$ ($1$ is always a square so the
\mathsc{lqnr} is larger than~$2$).
For any generator of~$\Z_{p}$, quadratic non-residues, as well as
their inverses ($s$ is invertible as it is non-zero and~$p$ is prime),
have an odd discrete logarithm.
Therefore the multiplication of~$k$ and the inverse of the \mathsc{lqnr} must be a square~$a^{2}$.
This means that the relation~${k=a^{2}\bigr(1+b^{2}\bigl)=a^{2}+{(ab)}^{2}}$ holds.

Now for the running time, under the \mathsc{erh},
\cite[Theorem~6.35]{Wedeniwski:2001:lqnr} shows that the \mathsc{lqnr}
should be lower than~${3\log^{2}(p)/2-44\log(p)/5+13}$.
From this, the expected number of Legendre symbol computations
is~$O\bigr(\log^{2}(p)\bigl)$ and this dominates the modular square
root computations.
\end{proof}

\begin{remark}
Another possibility is to use randomization: instead of using the
lowest quadratic non-residue (\mathsc{lqnr}), randomly select a
non-residue~$s$, and then decrement
it until~${s-1}$ is a quadratic residue ($1$ is a square so this will
terminate).
In practice, the running time seems very close to that of
\cref{alg:sosmodp} anyway, see, e.g.\ the implementation in
Givaro rev.~7bdefe6, \url{https://github.com/linbox-team/givaro}.
Also, when computing~$t$ sum of squares modulo the same prime, one can
compute the \mathsc{lqnr}
\emph{only once} to get all the sum of squares with an expected cost bounded
by~${\widetilde{O}\bigl({{\log^{3}}(p)+t{\log^{2}}(p)\bigr)}}$.
\end{remark}
\begin{remark}\label{alg:FFSoS}
Except in characteristic~$2$ or in algebraic closures, where every
element is a square anyway, \cref{alg:sosmodp} is easily extended over
any finite field: compute the \mathsc{lqnr} in the base prime field,
then use Tonelli-Shanks or Cipolla-Lehmer algorithm to compute square
roots in the extension field.

Denote by~$\FFSoS{q}{k}$ this algorithm decomposing~$k$ as a sum of
squares within any finite field~$\F_{q}$.
This is not always possible over infinite fields, but there
\cref{alg:sosmodp} still works anyway for the special case~${k=-1}$:
just run it in the prime sub-field, since $-1$ must be in it.
\end{remark}

\subsection[Finite fields with even extension and
conjugation]{Finite fields with even extension and
  $\MIAH{A}=\HermitianTranspose{\mat{A}}$}\label{ssec:herkeven}
With $\MIAH{A}=\HermitianTranspose{\mat{A}} $, we need a matrix~$\mat{Y}$ such
that~$\MatrixProduct{\mat{Y}}{\HermitianTranspose{Y}}={\MatrixProduct{\mat{Y}}{\CTranspose{Y}}=-\IdentityMatrix{}}$.
This is not possible anymore over the complex field, but works for any
even extension field, thanks to \cref{alg:sosmodp}. To see this, we
consider next the finite field $\F_{q^{2}}$, where $q$ is a power
of an arbitrary prime. Given $a\in\F_{q^{2}}$, we adopt the convention
that conjugation is given by the Frobenius automorphism:
\begin{equation}\Conjugation{a}=a^q.\end{equation}
The bar operator is $\F_q$-linear and has order $2$ on $\F_{q^{2}}$.

First, if~$-1$ is a square in~$\F_{q}$,
then~${Y=\sqrt{-1}\cdot\IdentityMatrixOfSize{n}}$ works in
$\F_{q^{2}}$ since then $\Conjugation{\sqrt{-1}}=\sqrt{-1}$:
$\MatrixProduct{\mat{Y}}{\CTranspose{Y}}=\sqrt{-1}\cdot\IdentityMatrixOfSize{n}\sqrt{-1}\cdot\IdentityMatrixOfSize{n}=-\IdentityMatrixOfSize{n}$.

Second, otherwise, ${{q}\equiv{{3}\bmod{4}}}$ and then there exists a
square root~$i$ of~$-1$ in~$\F_{q^{2}}$, from \cref{lem:lemsqrt}.
Further, one can build~$(a,b)$, both in the base field~$\F_{q}$, such
that~${a^{2}+b^{2}=-1}$, from~\cref{alg:sosmodp}.
Finally~${Y=(a+ib)\cdot{}{\IdentityMatrixOfSize{n}}}$
in~${{\F_{q^{2}}}^{n{\times}n}}$ is skew-unitary:
indeed, since~${{q}\equiv{{3}\bmod{4}}}$, we have
that~$i^q=i^{3+4k}=i^3(-1)^{2k}=-i$ and, therefore,
${\overline{a+ib}={(a+ib)}^{q}={a-ib}}$.
Finally
$\MatrixProduct{\mat{Y}}{\CTranspose{Y}}=(a+ib)(a-ib)\cdot\IdentityMatrixOfSize{n}=-\IdentityMatrixOfSize{n}$.

\subsection[Any field with positive characteristic and
conjugation]{Any field with positive characteristic and
  $\MIAH{A}=\HermitianTranspose{\mat{A}}$}\label{ssec:allherk}
If $-1$ is a square in the base field, or within an even extension we
have seen in~\cref{ssec:herkeven} that there exists diagonal
skew-unitary matrices. Otherwise, one can always resort to tridiagonal
ones as in~\cref{ssec:tridiag}.
For this, one can always build $(a,b)$ in the base field such that
$a^2+b^2=-1$ using~\cref{lem:pigeonhole}. Then,
$Y=\begin{smatrix}a & b\\-b &
  a\end{smatrix}\tensorproduct{\IdentityMatrixOfSize{n}}$ is a
skew-unitary matrix.
Indeed, since $a$ and $b$ live in the base field, they are invariant
by the Frobenius automorphism.
Therefore, $\CTranspose{Y}=\CTranspose{\begin{smatrix}a & b\\-b &
    a\end{smatrix}}\tensorproduct{\IdentityMatrixOfSize{n}}=\begin{smatrix}a & -b\\b & a\end{smatrix}\tensorproduct{\IdentityMatrixOfSize{n}}$
and
$\MatrixProduct{\mat{Y}}{\CTranspose{Y}}=(a^2+b^2)\cdot{\IdentityMatrixOfSize{n}}=-\IdentityMatrixOfSize{n}$.

\section{Towards a minimality result on the number of multiplications}
\label{sec:minimality}
Our~\cref{alg:MIAHMM} computes the product of a matrix over a ring by its (hermitian) transpose using only~$5$ block multiplications and the (hermitian) transpose of one of these block multiplications. Here, we use consider some vector-spaces and thus, restrict ourselves to consider matrices over a field.
\par
We reformulate in this section the method introduced by de~Groote in~\cite{groote:1978} in order to prove that the tensor rank of the~$\matrixsize{2}{2}$ matrix product is~$7$.
This method is used to prove the following result:
\begin{theorem}\label{thm:no4product}
There is no algorithm derived from non-commutative block~$\matrixsize{2}{2}$ matrix product algorithms that computes the product of a matrix over a field by its (hermitian) transpose using only~$4$ block multiplications and the (hermitian) transpose of one of these block multiplications.
\end{theorem}
This result does not state that it is never possible to multiply by the adjoint using fewer than~$5$ multiplications as shown by the following remark.
\begin{remark}\label{rk:4mult}
Over any ring with a square root~$i$ of~$-1$, there is a computational scheme requiring~$4$ multiplications and computing the product of a~$\matrixsize{2}{2}$-matrix by its transpose:
\begin{equation}
\begin{smatrix}
a & b \\
c & d
\end{smatrix}
\begin{smatrix}
a & c \\
b & d
\end{smatrix}
=
\begin{smatrix}
(a+ib)(a-ib) & \times \\
ac+bd & (c+id)(c-id)
\end{smatrix}
=
\begin{smatrix}
a^2+b^2 & \times \\
ac+bd & c^2+d^2
\end{smatrix}\!.
\end{equation}
This is the case for instance over~$\F_2$, where~${i=1}$, or over the complex numbers.
As this scheme requires for instance that~${aib=iba}$, at least some
commutativity is required, thus in general it does not apply to block
matrices and it is therefore not in the scope of~\cref{thm:no4product}.
\end{remark}
The following section is devoted to shortly present the framework used in this part of our work.
\subsection{The framework of bilinear maps encoded by tensors}
We present de~Groote's proof using a tensorial presentation of
bilinear maps; we recall briefly this standpoint through the following
well-known example of seven multiplications and we refer to~\cite{Landsberg:2016ab} for a complete introduction to this framework.
\begin{example}
Considered as~${\matrixsize{2}{2}}$ matrices, the matrix product~${\mat{C}=\MatrixProduct{A}{B}}$ could be computed using Strassen algorithm by performing the following computations (see~\cite{Strassen:1969:GENO}):
\begin{equation}
\label{eq:StrassenMultiplicationAlgorithm}
\begin{array}{ll}
\mathcolor{\triadone}{\rho_{1}}\leftarrow{\mathcolor{\triadone}{a_{11}}(\mathcolor{\triadone}{b_{12}-b_{22}})},
&
\\
\mathcolor{\triadtwo}{\rho_{2}}\leftarrow{(\mathcolor{\triadtwo}{a_{11}+a_{12}})\mathcolor{\triadtwo}{b_{22}}},
&
\mathcolor{\triadfour}{\rho_{4}}\leftarrow{(\mathcolor{\triadfour}{a_{12}-a_{22}})(\mathcolor{\triadfour}{b_{21}+b_{22}})},
\\
\mathcolor{\triadthree}{\rho_{3}}\leftarrow{(\mathcolor{\triadthree}{a_{21}+a_{22}}) \mathcolor{\triadthree}{b_{11}}},
&
\mathcolor{\triadfive}{\rho_{5}}\leftarrow{(\mathcolor{\triadfive}{a_{11}+a_{22}})(\mathcolor{\triadfive}{b_{11}+b_{22}})},
\\
\mathcolor{\triadsix}{\rho_{6}}\leftarrow{\mathcolor{\triadsix}{a_{22}}(\mathcolor{\triadsix}{b_{21}-b_{11}})},
&
\mathcolor{\triadseven}{\rho_{7}}\leftarrow{(\mathcolor{\triadseven}{a_{21}-a_{11}})(\mathcolor{\triadseven}{b_{11}+b_{12}})},
\\[\medskipamount]
\multicolumn{2}{c}{
\begin{mmatrix} c_{11} &c_{12} \\ c_{21} &c_{22} \end{mmatrix}
=
\begin{mmatrix}
\mathcolor{\triadfive}{\rho_{5}} + \mathcolor{\triadfour}{\rho_{4}} - \mathcolor{\triadtwo}{\rho_{2}} + \mathcolor{\triadsix}{\rho_{6}} &
\mathcolor{\triadsix}{\rho_{6}} + \mathcolor{\triadthree}{\rho_{3}} \\
\mathcolor{\triadtwo}{\rho_{2}} + \mathcolor{\triadone}{\rho_{1}}&
\mathcolor{\triadfive}{\rho_{5}} + \mathcolor{\triadseven}{\rho_{7}} + \mathcolor{\triadone}{\rho_{1}}- \mathcolor{\triadthree}{\rho_{3}}
\end{mmatrix}\!.}
\end{array}
\end{equation}
With~${m,n,p}$ equal to~$2$, this algorithm encodes a bilinear map:
\begin{equation}\label{eq:mxnTimesnxp}
\begin{array}{ccl}
\Matrices{\F}{\matrixsize{m}{n}} \times \Matrices{\F}{\matrixsize{n}{p}} & \rightarrow &\Matrices{\F}{\matrixsize{m}{p}}, \\
(\mat{A},\mat{B}) &\rightarrow & \MatrixProduct{A}{B}.
\end{array}
\end{equation}
We keep the indices~${m,n,p}$ in this section for the sake of clarity in order to distinguish the different spaces involved in the sequel.
The spaces~${\Matrices{\F}{\matrixsize{\cdot}{\cdot}}}$ can be endowed with the Frobenius product~${{\langle \mat{M},\mat{N}\rangle}={\Trace({\MatrixProduct{\Transpose{M}}{\mat{N}}})}}$ that establishes an isomorphism between~$\Matrices{\F}{\matrixsize{\cdot}{\cdot}}$ and its dual space~$\bigl(\Matrices{\F}{\matrixsize{\cdot}{\cdot}}\bigr)^{\star}$;
hence, it allows for example to associate the trilinear form~${\Trace(\MatrixProduct{\Transpose{C}}{\MatrixProduct{A}{B}})}$ and the matrix multiplication~(\ref{eq:mxnTimesnxp}):
\begin{equation}
\label{eq:TrilinearForm}
\Contraction{\tensor{S}}{}{3}:
\begin{array}[t]{ccc}
{\F}^{\matrixsize{m}{n}} \times {\F}^{\matrixsize{n}{p}} \times {({\F}^{\matrixsize{m}{p}})}^{\star}&\rightarrow & {\F}, \\
(\mat{A},\mat{B},\Transpose{C}) &\rightarrow & \langle \mat{C},\MatrixProduct{A}{B}\rangle.
\end{array}
\end{equation}
As by construction, the space of trilinear forms is the canonical dual space of order three tensor products, we could encode the Strassen multiplication algorithm~(\ref{eq:StrassenMultiplicationAlgorithm}) as the tensor~$\tensor{S}$ defined by:
\begin{equation}
\label{eq:StrassenTensor}
\begin{array}{r}
\tensor{S}=\sum_{i=1}^{7}{\Sigma^{i}_{1}}\!\tensorproduct\!{\Sigma^{i}_{2}}\!\tensorproduct\!{S^{3}_{i}}=
{\Sigma^{i}_{1}}\!\tensorproduct\!{\Sigma^{i}_{2}}\!\tensorproduct\!{S^{3}_{i}}=
\mathcolor{\triadone}{\begin{smatrix}1&0\\0&0\\\end{smatrix}\!\tensorproduct\!\begin{smatrix}0&1\\0&-1\\\end{smatrix}\!\tensorproduct\!\begin{smatrix}0&0\\1&1\\\end{smatrix}
}
\!+\!\\[\bigskipamount]
\mathcolor{\triadtwo}{\begin{smatrix}1&1\\0&0\\\end{smatrix}\!\tensorproduct\!\begin{smatrix}0&0\\0&1\\\end{smatrix}\!\tensorproduct\!\begin{smatrix}-1&0\\1&0\\\end{smatrix}}
\!+\!
\mathcolor{\triadthree}{\begin{smatrix}0&0\\1&1\\\end{smatrix}\!\tensorproduct\!\begin{smatrix}1&0\\0&0\\\end{smatrix}\!\tensorproduct\!\begin{smatrix}0&1\\0&-1\end{smatrix}}
\!+\!
\mathcolor{\triadfour}{\begin{smatrix}0&1\\0&-1\\\end{smatrix}\!\tensorproduct\!\begin{smatrix}0&0\\1&1\\\end{smatrix}\!\tensorproduct\!\begin{smatrix}1&0\\0&0\\\end{smatrix}}
\!+\!\\[\bigskipamount]
\mathcolor{\triadfive}{{\begin{smatrix}1&0\\0&1\end{smatrix}}\!\tensorproduct\!{\begin{smatrix}1&0\\0&1\end{smatrix}}\!\tensorproduct\!\begin{smatrix}1&0\\0&1\\\end{smatrix}}
\!+\!
\mathcolor{\triadsix}{\begin{smatrix}0&0\\0&1\\\end{smatrix}\!\tensorproduct\!\begin{smatrix}-1&0\\1&0\\\end{smatrix}\!\tensorproduct\!\begin{smatrix}1&1\\0&0\\\end{smatrix}}
\!+\!
\mathcolor{\triadseven}{\begin{smatrix}-1&0\\1&0\\\end{smatrix}\!\tensorproduct\!\begin{smatrix}1&1\\0&0\\\end{smatrix}\!\tensorproduct\!\begin{smatrix}0&0\\0&1\\\end{smatrix}}
\!
\end{array}
\end{equation}
in~${{({\F}^{\matrixsize{m}{n}})}^{\star} \tensorproduct {({\F}^{\matrixsize{n}{p}})}^{\star} \tensorproduct {\F}^{\matrixsize{m}{p}}}$ with~${m=n=p=2}$.
\end{example}
Remark that---as introduced in the above~\Cref{eq:StrassenTensor}---we are going to
use in the sequel the \emph{Einstein summation convention} in order to
simplify the forthcoming notations (according to this convention, when an
index variable appears twice in a term and is not otherwise defined,
it represents in fact the sum of that term over all the values of the index).

Starting from the tensor representation~$\tensor{S}$ of our algorithm, we could consider several \emph{contractions} that are the main objects manipulated in the sequel.
\subsection{Flattening tensors and isotropies}
The \emph{complete} contraction~${\Contraction{\tensor{S}}{}{3}({\mat{A}}\tensorproduct{\mat{B}}\tensorproduct{\Transpose{C}})}$ is defined as the following map:
\begin{equation}
\label{eq:Completecontraction}
\begin{array}{c}
{\left({\Dual{\Matrices{\F}{\matrixsize{m}{n}}}}\tensorproduct{\Dual{\Matrices{\F}{\matrixsize{n}{p}}}}\tensorproduct{\Matrices{\F}{\matrixsize{m}{p}}}\right)}
\tensorproduct
{\left({\Matrices{\F}{\matrixsize{m}{n}}}\tensorproduct{\Matrices{\F}{\matrixsize{n}{p}}}\tensorproduct{\Dual{\Matrices{\F}{\matrixsize{m}{p}}}}\right)}
\rightarrow
\F, \\[\smallskipamount]
{\left({\Sigma^{i}_{1}}\!\tensorproduct\!{\Sigma^{i}_{2}}\!\tensorproduct\!{S^{3}_{i}}\right)}
\tensorproduct
({\mat{A}}\tensorproduct{\mat{B}}\tensorproduct{\Transpose{C}}) \rightarrow
\langle {\Sigma^{i}_{1}}, \mat{A} \rangle
\langle {\Sigma^{i}_{2}}, \mat{B} \rangle
\langle {S^{3}_{i}},\Transpose{C}  \rangle.
\end{array}
\end{equation}
We already saw informally in the previous section that this complete contraction is~${\Trace(\MatrixProduct{\MatrixProduct{A}{B}}{C})}$ and we recall in the following remark some of its basic properties.
\begin{remark}
Given three invertible matrices:
\begin{equation}
\SymFirstTriad\in\Matrices{\F}{\matrixsize{m}{m}},\quad
\SymSecondTriad\in\Matrices{\F}{\matrixsize{p}{p}},\quad
\SymThirdTriad\in\Matrices{\F}{\matrixsize{n}{n}}
\end{equation}
that encodes changes of basis, the trace~${\Trace(\MatrixProduct{A}{\MatrixProduct{B}{C}})}$ is equal to:
\begin{equation}
\begin{array}{l}
\label{eq:isotropy}
 \Trace\bigl(\Transpose{(\MatrixProduct{\MatrixProduct{A}{B}}{C})}\bigr)
 =\Trace(\MatrixProduct{C}{\MatrixProduct{A}{B}})
=\Trace(\MatrixProduct{B}{\MatrixProduct{C}{A}}),\\
\textrm{and}\ \Trace\bigl(\MatrixProduct{\Inverse{\SymFirstTriad}}{\MatrixProduct{A}{\SymSecondTriad}}
 \cdot \Inverse{\SymSecondTriad} \cdot \mat{B} \cdot {\SymThirdTriad} \cdot \Inverse{\SymThirdTriad} \cdot \mat{C} \cdot {\SymFirstTriad}\bigr).
\end{array}
\end{equation}
\end{remark}
These relations illustrate the following theorem:
\begin{theorem}[{\cite[\S~2.8]{groote:1978a}}]
The isotropy group of the~$\matrixsize{n}{n}$ matrix multiplication tensor is~${{{\mathsc{psl}^{\pm}({\F^{n}})}^{\times 3}}\!\rtimes{\mathfrak{S}_{3}}}$, where~$\mathsc{psl}$ stands for the group of matrices of determinant~${\pm{1}}$ and~$\mathfrak{S}_{3}$ for the symmetric group on~$3$ elements.
\end{theorem}
The following classical statement redefines the \emph{sandwiching}
isotropy on a matrix multiplication tensor:
\begin{definition}\label{def:sandwiching}
Given~${\Isotropy{g}={(\SymFirstTriad\times\SymSecondTriad\times\SymThirdTriad)}}$ in~${\mathsc{psl}^{\pm}({\F}^{n})}^{\times 3}$, its action~${\IsotropyAction{\Isotropy{g}}{\tensor{S}}}$ on a tensor~$\tensor{S}$ is given by~${\IsotropyAction{\Isotropy{g}}{({\Sigma^{i}_{1}}\tensorproduct{\Sigma^{i}_{2}}\tensorproduct{S^{3}_{i}})}}$ where each summands is equal to:
\begin{equation}
\label{eq:sandwiching}
{\left(\MatrixProduct{\InvTranspose{\SymFirstTriad}}{\MatrixProduct{\Sigma^{i}_{1}}{\Transpose{\SymSecondTriad}}}\right)}
\tensorproduct
{\left(\MatrixProduct{\InvTranspose{\SymSecondTriad}}{\MatrixProduct{\Sigma^{i}_{2}}{\Transpose{\SymThirdTriad}}}\right)}
\tensorproduct
{\left(\MatrixProduct{\InvTranspose{\SymThirdTriad}}{\MatrixProduct{S^{3}_{i}}{\Transpose{\SymFirstTriad}}}\right)},\
\forall i\ \textrm{fixed}.
\end{equation}
\end{definition}
These isotropies will be used later; for the moment, let us now focus
our attention on the very specific standpoint on which is based the
forthcoming developments: flattenings.
\begin{definition}
Given a tensor~$\tensor{S}$, the third flattening~$\Contraction{\tensor{S}}{1}{3}$ (a.k.a.\ third~$1$-contraction) of the tensor~$\tensor{S}$ is:
\begin{equation}
\label{eq:ThirdFlattening}
\Contraction{\tensor{S}}{1}{3} :
\begin{array}[t]{ccc}
\Matrices{\F}{\matrixsize{m}{p}}&\rightarrow&{\Dual{\Matrices{\F}{\matrixsize{m}{n}}}}\tensorproduct{\Dual{\Matrices{\F}{\matrixsize{n}{p}}}},
\\[\smallskipamount]
\mat{M} &\rightarrow&{\langle \mat{M},S^{3}_{i}\rangle}\,{{\Sigma^{i}_{1}}\tensorproduct{\Sigma^{i}_{2}}}.
\end{array}
\end{equation}
\end{definition}
\begin{example}
To illustrate this definition and some important technicalities, let
us consider~$\Image{\Contraction{\tensor{S}}{1}{3}}$ the image of the
Strassen tensor~(\ref{eq:StrassenTensor}) flattening: this is a subspace of~${{\Dual{\Matrices{\F}{\matrixsize{m}{n}}}} \tensorproduct {\Dual{\Matrices{\F}{\matrixsize{n}{p}}}}}$.
More precisely, let us first consider only the fifth summand in~\cref{eq:StrassenTensor} and the image of its third flattening:
\begin{equation}\label{eq:FlatteningExample}
\Image{\Contraction{\begin{smatrix}1&0\\0&1\\\end{smatrix}\!\tensorproduct\!\begin{smatrix}1&0\\0&1\\\end{smatrix}\!\tensorproduct\!\begin{smatrix}1&0\\0&1\\\end{smatrix}}{1}{3}}=
\begin{smatrix} c_{11}+c_{22}&0&0&c_{11}+c_{22}\\ 0&0&0&0 \\ 0&0&0&0 \\ c_{11}+c_{22}&0&0&c_{11}+c_{22} \end{smatrix}
\ \forall (c_{11},c_{22})\in\F^{2}.
\end{equation}
The indeterminates~$c_{11}$ and~$c_{22}$ keep track of the domain of the flattening.
The~$\matrixsize{4}{4}$ right-hand side matrix in above expression should not be confused with the Kronecker product~${{\Identity_{\matrixsize{2}{2}}}\tensorproduct{\Identity_{\matrixsize{2}{2}}}}$ involved in the left-hand side.
In fact, the result~$\Identity_{\matrixsize{4}{4}}$ of this Kronecker product is of classical matrix rank~$4$ while the rank in~${{\Dual{\Matrices{\F}{\matrixsize{m}{n}}}} \tensorproduct {\Dual{\Matrices{\F}{\matrixsize{n}{p}}}}}$ of the right-hand side matrix~(\ref{eq:FlatteningExample}) is~$1$ by construction.
Hence, even if we present in this section the elements of~${\Matrices{\F}{\matrixsize{\cdot}{\cdot}}}$ as matrices, we use a vectorization (e.g.~${{\Matrices{\F}{\matrixsize{m}{n}}}\simeq{\Matrices{\F}{mn}}}$) of these matrices in order to perform correctly our computations.
Taking this standpoint into account we obtain the following
description of the whole Strassen third flattening image as:
\begin{equation}
\Image{\Contraction{S}{1}{3}}=
\begin{mmatrix}
c_{11} & c_{12} & 0 & 0 \\
     0 &      0 & c_{11} & c_{12} \\
c_{21} & c_{22} & 0 & 0 \\
     0 &      0 & c_{21} & c_{22}
\end{mmatrix}
\end{equation}
that could be guessed almost without computation.
In fact, this right-hand side matrix is just the matrix of the bilinear form defining the trilinear encoding of the matrix product:
\begin{equation}
\Trace(\MatrixProduct{\MatrixProduct{A}{B}}{\Transpose{C}})
=
\begin{mmatrix}
a_{11} & a_{12} & a_{21} & a_{22}
\end{mmatrix}
\begin{mmatrix}
c_{11} & c_{12} & 0 & 0 \\
     0 &      0 & c_{11} & c_{12} \\
c_{21} & c_{22} & 0 & 0 \\
     0 &      0 & c_{21} & c_{22}
\end{mmatrix}
\begin{mmatrix}
b_{11} \\ b_{12} \\ b_{21} \\ b_{22}
\end{mmatrix}\!.
\end{equation}
\end{example}
Hence, the flattening~$\Image{\Contraction{S}{1}{3}}$ is a canonical description of the matrix product independent from the algorithm/tensor used to encode it; in particular, it is an invariant under the action of the isotropies introduced in~\cref{def:sandwiching}. We are going to use these objects and their properties in the following section.
\subsection{Presentation of de Groote's method}
We are interested in a situation where, given a bilinear map, a classical representation by a tensor~$\tensor{C}$ is known and we wish to disprove the existence of a tensor representation of a given rank.
Inspired by Steinitz exchange theorem, de~Groote introduced in~\cite[\S~1.5]{groote:1978} the following definition to handle this issue.
\begin{definition}
Given a tensor~$\tensor{T}$ encoding a bilinear map whose codomain is~$\BilinearMapCoDomain$ and given~$q$ rank-one elements~${\tensor{P}^{i}}$, linearly independent in~${{\Dual{\Matrices{\F}{\matrixsize{m}{n}}}} \tensorproduct {\Dual{\Matrices{\F}{\matrixsize{n}{p}}}}}$, let us denote by~${\BilinearMapCoDomainSubSpace(\Image{\Contraction{T}{1}{3}},\tensor{P}^{1},\ldots,\tensor{P}^{q})}$ the linear subspace of~$\BilinearMapCoDomain$ defined by:
\begin{equation}
\label{eq:PartiallyKnownRepresentation}
\textup{LinearSpan}
\left\lbrace
\begin{array}{c}
\mat{M}\in\BilinearMapCoDomain \mid \exists{(u_{1},\ldots,u_{q})}\in{{\F}^{q}},\\
 \textup{Rank}_{{\Dual{\Matrices{\F}{\matrixsize{m}{n}}}} \tensorproduct {\Dual{\Matrices{\F}{\matrixsize{n}{p}}}}}
\left( \Contraction{\tensor{T}}{1}{3}(\mat{M}) + u_{j} \tensor{P}^{j} \right)= 1
\end{array}
\right\rbrace.
\end{equation}
\end{definition}
We introduced the notation~$\BilinearMapCoDomain$ for the codomain of the considered bilinear map in order to highlight
the fact that it is isomorphic---via the Frobenius isomorphism---to the domain of the flattening and to show how it is used in the following.
\par
The following proposition allows to construct an effective test that checks if there exists a tensor of rank~${\dim{\BilinearMapCoDomain}+q}$ that defines the considered bilinear map.
\begin{proposition}\label{prop:deGrooteMainTrick}
If there exists a tensor~$\tensor{T}$ of rank~${\dim{\BilinearMapCoDomain}+q}$ encoding a bilinear map with codomain~$\BilinearMapCoDomain$ then there are~$q$ rank-one elements~${\tensor{P}^{i}}$ linearly independent in~${{\Dual{\Matrices{\F}{\matrixsize{m}{n}}}} \tensorproduct {\Dual{\Matrices{\F}{\matrixsize{n}{p}}}}}$ such that~${\BilinearMapCoDomainSubSpace(\Image{\Contraction{\tensor{T}}{1}{3}},\tensor{P}^{1},\ldots,\tensor{P}^{q})}$ is~$\BilinearMapCoDomain$.
\end{proposition}
\begin{proof}
Let us assume that there exists a tensor~$\tensor{C}$ encoding the considered bilinear map, that its tensor rank~$\rho$ is greater than~${\dim{\BilinearMapCoDomain}+q}$ and that it is defined by the sum~${{\tensor{Q}^{i}}\tensorproduct{R_{i}}}$. Remark that the set~$(R_{i})_{{1}\leq{i}\leq{\rho}}$ is a generating set of the space~$\BilinearMapCoDomain$ by hypothesis.
\par
Suppose now that there exists a tensor~$\tensor{T}$ of rank~${{r}={\dim{\BilinearMapCoDomain}+q}}$ encoding with fewer summands the same considered bilinear map:
\begin{equation}
\tensor{T} = {\tensor{P}^{k}}\tensorproduct{S_{k}},
\quad
{S_{k}}\in{\BilinearMapCoDomain}\subset{\Matrices{\F}{\matrixsize{m}{p}}},
\quad
{\tensor{P}^{k}}\in{{\Dual{\Matrices{\F}{\matrixsize{m}{n}}}} \tensorproduct {\Dual{\Matrices{\F}{\matrixsize{n}{p}}}}}.
\end{equation}
The elements~$\tensor{P}^{k}$ are linearly independent, otherwise~$\tensor{T}$ could be expressed with even fewer terms.
Furthermore, there is a subset of~${(S_{i})_{{1}\leq{i}\leq{\dim{\BilinearMapCoDomain}}}}$ that is a base of~${\BilinearMapCoDomain}$ (otherwise,~$\tensor{T}$ could not encode the same bilinear map as~$\tensor{C}$).
\par
Suppose that we reorder this set so that the base is~${(S_{i+q})_{{1}\leq{i}\leq{\dim{\BilinearMapCoDomain}}}}$.
By invariance of the flattening map, the following relations hold:
\begin{equation}
\forall\ {\mat{M}}\in{\BilinearMapCoDomain},\quad
\Contraction{\tensor{C}}{1}{3}(\mat{M})=\langle \mat{M}, R_{k}\rangle \tensor{Q}^{k} =
\Contraction{\tensor{T}}{1}{3}(\mat{M})=\langle \mat{M}, S_{k}\rangle \tensor{P}^{k}.
\end{equation}
By introducing a base~$(B^{i})_{{1}\leq{i}\leq{\dim{\BilinearMapCoDomain}}}$ of~$\BilinearMapCoDomain$, we could summarize this situation under a matricial standpoint as follows:
\begin{equation}\label{eq:CanonicalWRTTensor}
\begin{mmatrix}
\tensor{P}^{1} \\
\vdots \\
\tensor{P}^{q} \\
\langle B^{1},R_{k}\rangle \tensor{Q}^{k} \\
\vdots \\
\langle B^{\dim{\BilinearMapCoDomain}},R_{k}\rangle \tensor{Q}^{k}
\end{mmatrix}
=
\left(
\begin{array}{cc}
 \IdentityMatrixOfSize{\matrixsize{q}{q}} & 0 \\
C & D
\end{array}
\right)
\left(
\begin{array}{c}
\tensor{P}^{1} \\
\vdots \\
\tensor{P}^{q} \\
\tensor{P}^{q+1} \\
\vdots \\
\tensor{P}^{q+\dim{\BilinearMapCoDomain}}
\end{array}
\right)
\end{equation}
where the matrices~$C$ and~$D$ are such that:
\begin{equation}
\forall\ i\in\{1,\ldots,\dim{\BilinearMapCoDomain}\},%
\begin{array}{ll}
C_{ij}=\langle B^i,S_j \rangle,& \forall\ j\in\{1,\ldots,q\},\\
D_{ij}=\langle B^i,S_{q+j} \rangle,& \forall\ j\in\{1,\ldots,\dim{\BilinearMapCoDomain}\}.
\end{array}
\end{equation}
As, by hypothesis,~$(S_{i+q})_{{1}\leq{i}\leq{\dim{\BilinearMapCoDomain}}}$ is a basis of~${\BilinearMapCoDomain}$, the matrix~$D$ is invertible and we could rewrite~\cref{eq:CanonicalWRTTensor} as follows:
\begin{equation}\label{eq:InducedWRTTensor}
\begin{array}{l}
U=-\MatrixProduct{D^{-1}}{C},\\  V=D^{-1},
\end{array}
\!\!\left(
\begin{array}{cc}
 \IdentityMatrixOfSize{\matrixsize{q}{q}} & 0 \\
 U & V
\end{array}
\right)\!\!
\left(\!\!
\begin{array}{c}
\tensor{P}^{1} \\
\vdots \\
\tensor{P}^{q} \\
\langle B^{1},R_{k}\rangle \tensor{Q}^{k} \\
\vdots \\
\langle B^{\dim{\BilinearMapCoDomain}},R_{k}\rangle \tensor{Q}^{k}
\end{array}
\!\!\right)
\!=\!
\left(\!\!
\begin{array}{c}
\tensor{P}^{1} \\
\vdots \\
\tensor{P}^{q} \\
\tensor{P}^{q+1} \\
\vdots \\
\tensor{P}^{q+\dim{\BilinearMapCoDomain}}
\end{array}
\!\!\right)\!.
\end{equation}
The~$\dim{\BilinearMapCoDomain}$ lines of the ${(U,V)}$ matrices
in~\Cref{eq:InducedWRTTensor} give
us $\dim{\BilinearMapCoDomain}$ vectors~$(u^{j}_{1},\dots,u^{j}_{q},v^{j}_{1},\ldots,v^{j}_{\dim{\BilinearMapCoDomain}})$ such that for all~$j$ in~${\{1,\ldots,{\dim{\BilinearMapCoDomain}}\}}$
\begin{multline}
\textup{Rank}_{{\Dual{\Matrices{\F}{\matrixsize{m}{n}}}} \tensorproduct {\Dual{\Matrices{\F}{\matrixsize{n}{p}}}}}\left(
u^{j}_{i}\tensor{P}^{i} + {v^{j}_{i}}\langle B^{i},R_{k}\rangle\tensor{Q}^{k}
\right)
=\textup{Rank}_{{\Dual{\Matrices{\F}{\matrixsize{m}{n}}}} \tensorproduct {\Dual{\Matrices{\F}{\matrixsize{n}{p}}}}}
\tensor{P}^{q+j}\\
= 1.
\end{multline}
To conclude, we remark that these last relations show that all the matrices~${v^{j}_{i}B^{i}}$ are in the subspace~${\BilinearMapCoDomainSubSpace(\Image{\Contraction{\tensor{T}}{1}{3}},\tensor{P}^{1},\ldots,\tensor{P}^{q})}$.
As they are~$\dim{\BilinearMapCoDomain}$ independent linear combinations of basis elements of~$\BilinearMapCoDomain$, these matrices form another of its bases and thus the subspace~${\BilinearMapCoDomainSubSpace(\Image{\Contraction{\tensor{T}}{1}{3}},\tensor{P}^{1},\ldots,\tensor{P}^{q})}$ is equal to~${\BilinearMapCoDomain}$.
\end{proof}
\subsection{Adaptation to the Hermitian case}
In order to use \cref{prop:deGrooteMainTrick} to prove~\cref{thm:no4product}, we have to show that for any element~${\tensor{P}={\mat{P}}\tensorproduct{\mat{Q}}}$ in~${{\Dual{\Matrices{\F}{\matrixsize{m}{n}}}} \tensorproduct {\Dual{\Matrices{\F}{\matrixsize{n}{p}}}}}$ the subspace~${\BilinearMapCoDomainSubSpace(\Image{\Contraction{\tensor{T}}{1}{3}},\tensor{P},\HermitianTranspose{\tensor{P}})}$ is not equal to~${\BilinearMapCoDomain}$ (with~${\HermitianTranspose{\tensor{P}}={\HermitianTranspose{Q}}\tensorproduct{\HermitianTranspose{P}}}$).
This vector-space~$\BilinearMapCoDomain$ is a~$3$ dimensional vector-space spanned by all the outputs of our bilinear map.
Let us start to make this strategy more precise by the following remark.
\begin{remark}
A classical block version of bilinear algorithm (e.g.~\cite[\S~6.3.1]{jgd:2008:toms}) computing the product of a matrix by its adjoint is:
\begin{equation}
\begin{smatrix}
\mat{A_{11}}&\mat{A_{12}} \\
\mat{A_{21}}&\mat{A_{22}}
\end{smatrix}
\begin{smatrix}
\HermitianTranspose{A_{11}}&\HermitianTranspose{A_{21}} \\
\HermitianTranspose{A_{12}}&\HermitianTranspose{A_{22}}
\end{smatrix}
=
\begin{smatrix}
\mat{A_{11}}\HermitianTranspose{A_{11}}+ \mat{A_{12}}\HermitianTranspose{A_{12}} & \times\\
\mat{A_{21}}\HermitianTranspose{A_{11}}+ \mat{A_{22}}\HermitianTranspose{A_{12}} &
\mat{A_{21}}\HermitianTranspose{A_{21}}+ \mat{A_{22}}\HermitianTranspose{A_{22}} \\
\end{smatrix}\!.
\end{equation}
As the result of this algorithm is self-adjoint, by~\cref{lem:uplow} there is no need to compute the top-right coefficient and thus, we conclude that the dimension of~$\BilinearMapCoDomain$ is at most~$3$.
Hence, there exists a bilinear map encoded by a tensor~${\tensor{H}={\Sigma^{i}_{1}}\tensorproduct{\Sigma^{i}_{2}}\tensorproduct{S^{3}_{i}}}$ of rank~$6$ that computes the product of a matrix by its hermitian transpose and the image of its third flattening~$\Contraction{\tensor{H}}{1}{3}$ is
\begin{equation}
\label{eq:ThirdFlatteningH}
\Image{\Contraction{\tensor{H}}{1}{3}}(\Matrices{\F}{\matrixsize{m}{p}})
=
\begin{smatrix}
c_{1} & 0     & 0     & 0 \\
0     & 0     & c_{1} & 0 \\
c_{2} & c_{3} & 0     & 0 \\
0     & 0     & c_{2} & c_{3}
\end{smatrix}\!.
\end{equation}
\end{remark}
We need a last standard definition in order to classify all possible tensors~$\tensor{P}$ considered in the sequel.
\begin{definition}\label{def:type}
Given a tensor~$\tensor{P}$ decomposable as sum of rank-one tensors:
\begin{equation}
\label{eq:5}
\tensor{P}=\sum_{i=1}^{q} \tensorproduct_{j=1}^{s} \mat{P_{ij}}\ \textrm{where}\ \mat{P_{ij}}\ \textrm{are matrices}.
\end{equation}
The list~${[{(\Rank{\mat{P_{ij}}})}_{j={{1}\ldots{s}}}]}_{i=1\ldots q}$ is called the \emph{type} of tensor~$\tensor{P}$.
\end{definition}
\begin{remark}\label{rem:ProductsClassification}
In our situation~$q$ is one,~$s$ is two and the~$\mat{P_{ij}}$ are~$\matrixsize{2}{2}$ matrices; hence, the tensor~$\tensor{P}$ could only have type~$[(1,1)]$, $[(1,2)]$, $[(2,1)]$ or~$[(2,2)]$.
\end{remark}
We also use the isotropies presented in~\cref{def:sandwiching} in order to simplify as much as possible the tensor~$\tensor{P}$ as illustrated in the proof of the following statement.
Furthermore, let us first introduce several notations:
\begin{itemize}
\item we denote by~${\Subs{x=0}{\mat{G}}}$ the matrix~$\mat{G}$ in which the indeterminate~$x$ is replaced by~$0$;
\item the determinant of the matrix~${\begin{smatrix} G_{i,j} & G_{i,l}\\ G_{k,j} & G_{k,l} \end{smatrix}}$ is denoted by~$\Minor{i,j}{k,l}{\mat{G}}$;
\item as the rank of a matrix is invariant under elementary row and
  columns operations, we also use the
  notation~${\AddRow{i}{j}{\ell}{\mat{G}}}$ for the matrix resulting
  from the addition to the~$i$th line of~$\mat{G}$ of its~$j$th line
  multiplied by~$\ell$.
\end{itemize}
\begin{lemma}\label{fct:1x}
There is no tensor~$\tensor{P}$ in~${{\Dual{\Matrices{\F}{\matrixsize{m}{n}}}} \tensorproduct {\Dual{\Matrices{\F}{\matrixsize{n}{p}}}}}$ of type~${(1,i)}$ with~$i$ equals to~$1$ or~$2$ such that the subspace~${\BilinearMapCoDomainSubSpace(\Contraction{\tensor{H}}{1}{3},\tensor{P},\HermitianTranspose{\tensor{P}})}$ is equal to~$\BilinearMapCoDomain$.
\end{lemma}
\begin{proof}
Let us consider a tensor~${\tensor{P}={\mat{A}}\tensorproduct{\mat{B}}}$ of type~$(1,i)$ with~${i=1,2}$.
As the first component~$\mat{A}$ is of rank one, there exists two vectors such that:
\begin{equation}
\mat{A}=
{\begin{mmatrix} a_{1} & a_{2}\end{mmatrix}}
\tensorproduct
{\begin{mmatrix} b_{1} \\ b_{2}\end{mmatrix}}
=
\begin{mmatrix}
a_{1}b_{1} & a_{2}b_{1} \\
a_{1}b_{2} & a_{2}b_{2}
\end{mmatrix}\!.
\end{equation}
If the coefficient~$a_{2}$ is zero, we choose a matrix~${\SymSecondTriad}$ as the identity matrix.
If the coefficient~$a_{1}$ is zero, we could consider a permutation
matrix~${\SymSecondTriad=\begin{smatrix}0&1\\1&0\end{smatrix}}$;
otherwise, if
$s=a_{1}\Conjugate{a_{1}}+a_{2}\Conjugate{a_{2}}\neq{0}$, consider:
\begin{equation}
\SymSecondTriad
=
\begin{mmatrix}
\Conjugate{a_{1}} & a_{2} \\
\Conjugate{a_{2}} &-a_{1}
\end{mmatrix}.
\end{equation}
Then we have both 
$\MatrixProduct{\SymSecondTriad}{\HermitianTranspose{\SymSecondTriad}}=s\cdot{\IdentityMatrixOfSize{\matrixsize{2}{2}}}$
and $\MatrixProduct{A}{\SymSecondTriad}=\begin{mmatrix} sb_{1} & 0 \\ sb_{2} & 0 \end{mmatrix}$.

Hence, in any of these cases, there always exists a matrix~$\SymSecondTriad$, which inverse is a multiple of
its hermitian transpose, such that the isotropy~$\Isotropy{g}$ defined by~${{({\IdentityMatrixOfSize{\matrixsize{m}{m}}}\times{\SymSecondTriad}\times{\IdentityMatrixOfSize{\matrixsize{n}{n}}})}}$ satisfies the following properties:
\begin{equation}\label{eq:type1x}
\IsotropyAction{\Isotropy{g}}
{\tensor{P}}=
{\begin{mmatrix} u_{1} & 0 \\ u_{2} & 0 \end{mmatrix}}
\tensorproduct
{\begin{mmatrix} v_{1} & v_{2} \\ v_{3} & v_{4} \end{mmatrix}},\quad
\IsotropyAction{\Isotropy{g}}{\HermitianTranspose{\tensor{P}}}=\frac{1}{s}\HermitianTranspose{(\IsotropyAction{\Isotropy{g}}{\tensor{P}})}
\end{equation}
With the above notations, conventions and isotropy's action, given any~$\mat{M}$ in~$\BilinearMapCoDomain$, the~$\matrixsize{4}{4}$ matrix~${\Contraction{\tensor{H}}{1}{3}(\mat{M})+y_{1}\tensor{P}+y_{2}\HermitianTranspose{\tensor{P}}}$ is:
\begin{equation}\small
\begin{mmatrix}
{y_{1}}{u_{1}}{v_{1}}+{y_{2}}{\HermitianTranspose{v_{1}}}{\HermitianTranspose{u_{1}}}+{c_{1}}&{y_{1}}{u_{1}}{v_{2}}+{y_{2}}{\HermitianTranspose{v_{1}}}{\HermitianTranspose{u_{2}}}&{y_{1}}{u_{1}}{v_{3}}&{y_{1}}{u_{1}}{v_{4}}\\
{y_{2}}{\HermitianTranspose{v_{3}}}{\HermitianTranspose{u_{1}}}&{y_{2}}{\HermitianTranspose{v_{3}}}{\HermitianTranspose{u_{2}}}&{c_{1}}&0\\
{y_{1}}{u_{2}}{v_{1}}+{y_{2}}{\HermitianTranspose{v_{2}}}{\HermitianTranspose{u_{1}}}+{c_{2}}&{y_{1}}{u_{2}}{v_{2}}+{y_{2}}{\HermitianTranspose{v_{2}}}{\HermitianTranspose{u_{2}}}+{c_{3}}&{y_{1}}{u_{2}}{v_{3}}&{y_{1}}{u_{2}}{v_{4}}\\
{y_{2}}{\HermitianTranspose{v_{4}}}{\HermitianTranspose{u_{1}}}&{y_{2}}{\HermitianTranspose{v_{4}}}{\HermitianTranspose{u_{2}}}&{c_{2}}&{c_{3}}
\end{mmatrix}\!.
\end{equation}
This matrix is supposed to be of rank one.
Thus, all its~$\matrixsize{2}{2}$ minors are equal to~$0$. Then,
either~$c_{1}$ or~$c_{3}$ is equal to~$0$ and for any
such~$\tensor{P}$ the
subspace~${\BilinearMapCoDomainSubSpace(\Contraction{\tensor{H}}{1}{3},\tensor{P},\HermitianTranspose{\tensor{P}})}$
is thus not~$\BilinearMapCoDomain$.
\par
There remains the case~${s=0}$.
Then let 
\begin{equation}
  \SymSecondTriad = \begin{mmatrix} {a_{1}}^{-1} & -a_{2} \\ 0 & a_{1} \end{mmatrix}.  
\end{equation}
The inverse of~$\SymSecondTriad$ is no longer related to~$\HermitianTranspose{\SymSecondTriad}$  anymore but~$\SymSecondTriad$ still transforms~$\mat{A}$ into a single non-zero column matrix:
$\MatrixProduct{A}{\SymSecondTriad}=\begin{mmatrix} b_{1} & 0 \\
  b_{2} & 0 \end{mmatrix}$.
Thus, for this~$\SymSecondTriad$, the action of the isotropy~${\Isotropy{g}=(\IdentityMatrixOfSize{\matrixsize{m}{m}}\times{\SymSecondTriad}\times{\IdentityMatrixOfSize{\matrixsize{n}{n}}})}$ is:
\begin{equation}\label{eq:orthog}
\IsotropyAction{\Isotropy{g}}{\tensor{P}}=
{\begin{mmatrix} b_{1} & 0 \\ b_{2} & 0 \end{mmatrix}}
\tensorproduct
{\begin{mmatrix} z_{11} & z_{12} \\ z_{21} & z_{22} \end{mmatrix}}\quad\textrm{and}\quad
\IsotropyAction{\Isotropy{g}}{\HermitianTranspose{\tensor{P}}}={\mat{U}}\tensorproduct{\mat{V}}.
\end{equation}
With the above notations, conventions and isotropy's action, given
any~$\mat{M}$ in~$\BilinearMapCoDomain$, the~$\matrixsize{4}{4}$
matrix~$\mat{G}={\Contraction{\tensor{H}}{1}{3}(\mat{M})+y_{1}\tensor{P}+y_{2}\HermitianTranspose{\tensor{P}}}$
is thus:
\begin{equation}\label{eq:szero} 
\begin{smatrix}
{y_{1}}{b_{1}}{z_{11}}+{y_{2}}{u_{11}}{v_{11}}+{c_1}&{y_{1}}{b_{1}}{z_{12}}+{y_{2}}{u_{11}}{v_{12}}&{y_{1}}{b_{1}}{z_{21}}+{y_{2}}{u_{11}}{v_{21} }&{y_{1}}{b_{1}}{z_{22}}+{y_{2}}{u_{11}}{v_{22}}\\ {y_{2}}{u_{12}}{v_{11}}&{y_{2}}{u_{12}}{v_{12}}&{y_{2}}{u_{12}}{v_{21}}+{c_1}&{y_{2}}{u_{12}}{v_{22}}\\ {y_{1}}{b_{2}}{z_{11}}+{y_{2}}{u_{21}}{v_{11}}+{c_2}&{y_{1}}{b_{2}}{z_{12}}+{y_{2}}{u_{21}}{v_{12}}+{c_3}&{y_{1}}{b_{2}}{z_{21}}+{y_{2}}{u_{21}}{v_{21}}&{y_{1}}{b_{2}}{z_{22}}+{y_{2}}{u_{21}}{v_{22}}\\ {y_{2}}{u_{22}} {v_{11}}&{y_{2}}{u_{22}}{v_{12}}&{y_{2}}{u_{22}}{v_{21}}+{c_2}&{y_{2}}{u_{22}} {v_{22}}+{c_3} 
\end{smatrix}\!.
\end{equation} 
This matrix is supposed to be of rank one in~${\BilinearMapCoDomainSubSpace(\Contraction{\tensor{H}}{1}{3},\tensor{P},\HermitianTranspose{\tensor{P}})}$.
Thus, all its~$\matrixsize{2}{2}$ minors are equal to~$0$.
\par
On the one hand, if~$y_{2}$ is zero, then the constraints of \Cref{eq:szero} show that~$\Minor{2,3}{4,4}{\Subs{y_{2}=0}{\mat{G}}}=c_{1}c_{3}$ is equal to~$0$.
On the other hand, if~$y_{2}$ is different from~$0$ then the minor~${\Minor{2,2}{4,4}{\mat{G}}}$ is equal to~$c_{3}y_{2}u_{12}v_{12}$ and supposed to be equal to zero by hypothesis. 
We also have that the minor~${\Minor{2,1}{4,4}{\mat{G}}}$ is equal to~$c_{3}y_{2}u_{12}v_{11}$ and is also supposed to be equal to zero by hypothesis. 
We are going to explore all the consequences induced by this constraint.
\begin{equation}\label{eq:startcase}
\small\begin{array}{llll}
u_{12}\neq 0, v_{12}\neq 0 & \multicolumn{3}{l}{\rightarrow c_{3}=0,} \\
u_{12}=0, v_{12}=0 & \multicolumn{3}{l}{\rightarrow \Minor{2,1}{4,3}{\Subs{u_{12}=0,v_{12}=0}{\mat{G}}} = y_{2}u_{22}v_{11}c_{1}=0,}\\
u_{12}=0, v_{12}=0, & u_{22}=0 & \multicolumn{2}{l}{\rightarrow \Minor{2,3}{4,4}{\Subs{u_{12}=0,v_{12}=0,u_{22}=0}{\mat{G}}} = c_{1}c_{3}=0,}\\
u_{12}=0, v_{12}=0, & v_{11}=0 & \text{see thereafter}\\
u_{12}=0, v_{12}\neq 0&\multicolumn{3}{l}{\rightarrow\Minor{2,2}{4,3}{\Subs{u_{12}=0}{\mat{G}}}=-y_{2} v_{12} u_{22} c_1 =0,}\\
u_{12}=0, v_{12}\neq 0,&u_{22}\neq 0& \multicolumn{2}{l}{\rightarrow c_1=0,}\\
u_{12}=0, v_{12}\neq 0,&u_{22}=0 & \multicolumn{2}{l}{\rightarrow \Minor{2,3}{4,4}{\Subs{u_{12}=0,u_{22}=0}{\mat{G}}}=c_1c_3=0.}
\end{array}
\end{equation}
Now, if~$u_{12}$ is different from~${0}$, then from the first two minors, either The relations~${v_{12}=v_{11}=0}$ hold or~${c_3=0}$.
\par
Further, not both $b_1$ and $b_2$ can be zero, otherwise the tensor is of
rank $0$. W.l.o.g., suppose that $b_2\neq{0}$
and let
$\mat{G}'=\AddRow{1}{3}{-b_1/b_2}{\Subs{v_{12}=0,v_{11}=0}{\mat{G}}}$.
Then $\Minor{1,2}{2,4}{\mat{G}'}=(-b_1/b_2)y_2c_3u_{12}v_{22}$ and
either $b_{1}=0$ or $v_{22}=0$ or $c_3=0$.
This gives the following distinctions (recall that now $u_{12}\neq{0}$ and $y_{2}\neq{0}$):
\begin{equation}
\begin{array}{llll}
b_{1}=0, & \multicolumn{3}{l}{\rightarrow\Minor{1,1}{2,4}{\Subs{b_{1}=0}{\mat{G}'}}=y_{2}u_{12}v_{22} c_1 =0,}\\
b_{1}=0, & v_{22}\neq{0} & \multicolumn{2}{l}{\rightarrow c_1 =0,}\\
b_{1}=0, & v_{22}=0 & \multicolumn{2}{l}{\rightarrow\Minor{1,1}{4,4}{\Subs{}{\mat{G}'}}=c_1c_3=0,}\\
b_{1}\neq{0},& v_{22}=0 & \multicolumn{2}{l}{\rightarrow\Minor{1,2}{4,4}{\Subs{}{\mat{G}'}}=(-b_{1}/b_{2}){c_3}^2=0,}\\
\end{array}
\end{equation}
There remains the case $u_{12}=0$, $v_{12}=0$ and  $v_{11}=0$ in \Cref{eq:startcase}.
Here also, w.l.o.g., suppose that $b_2\neq{0}$ and let
$\mat{G}^*=\AddRow{1}{3}{-b_1/b_2}{\Subs{u_{12}=0,v_{12}=0,v_{11}=0}{\mat{G}}}$.
\begin{equation}
\begin{array}{llll}
b_1=0 & \multicolumn{3}{l}{\rightarrow\Minor{1,1}{2,3}{\Subs{b_{1}=0}{\mat{G}^*}}={c_1}^2=0,}\\
b_1\neq{0} & \multicolumn{3}{l}{\rightarrow
  \Minor{1,2}{2,3}{\Subs{}{\mat{G}^*}}=(-b_{1}/b_{2})c_1c_3=0.}\\
\end{array}
\end{equation}
Thus, in any cases, for any such~$\tensor{P}$, the set~${\BilinearMapCoDomainSubSpace(\Contraction{\tensor{H}}{1}{3},\tensor{P},\HermitianTranspose{\tensor{P}})}$ is not~$\BilinearMapCoDomain$.
\par
Note that a computational way to see this, is to perform a Gr\"obner basis computation,
directly from \Cref{eq:szero}: for instance over $\CC$ this gives that
the relation ${c_1}^2{c_3}^2=0$ must hold and that the set is not the
full codomain.
\end{proof}
According to~\cref{rem:ProductsClassification}, the above computations deal with half of the cases to consider.
We remark that, mutatis mutandis, similar computations exclude also the existence of an algorithm where~$\tensor{P}$ is of type~$(2,1)$.
We could consider now the last case.
\begin{lemma}\label{fct:2x}
There is no tensor~$\tensor{P}$ in~${{\Dual{\Matrices{\F}{\matrixsize{m}{n}}}} \tensorproduct {\Dual{\Matrices{\F}{\matrixsize{n}{p}}}}}$ of type~${(2,2)}$ such that the subspace~${\BilinearMapCoDomainSubSpace(\Contraction{\tensor{H}}{1}{3},\tensor{P},\HermitianTranspose{\tensor{P}})}$ is equal to~$\BilinearMapCoDomain$.
\end{lemma}
\begin{proof}
First, let us consider a tensor~$\tensor{P}$ of type~$(2,2)$.
Thus, there exists~$\SymSecondTriad$ such that the action of the isotropy~${\Isotropy{g}=(\IdentityMatrixOfSize{\matrixsize{m}{m}}\times{\SymSecondTriad}\times\IdentityMatrixOfSize{\matrixsize{n}{n}})}$ is:
\begin{equation}
\IsotropyAction{\Isotropy{g}}{\tensor{P}}=
{\begin{mmatrix} 1 & 0 \\ 0 & 1 \end{mmatrix}}
\tensorproduct
{\begin{mmatrix} z_{11} & z_{12} \\ z_{21} & z_{22} \end{mmatrix}}\quad\textrm{and}\quad
\IsotropyAction{\Isotropy{g}}{\HermitianTranspose{\tensor{P}}}={\mat{U}}\tensorproduct{\mat{V}}.
\end{equation}
With the above notations, conventions and isotropy's action, given any~$\mat{M}$ in~$\BilinearMapCoDomain$, the~$\matrixsize{4}{4}$ matrix~$\mat{G}={\Contraction{\tensor{H}}{1}{3}(\mat{M})+y_{1}\tensor{P}+y_{2}\HermitianTranspose{\tensor{P}}}$ is:
\begin{equation}\label{eq:type2x}
\begin{smatrix}
y_{2}{u_{11}}{v_{11}}+{c_{1}}+y_{1}{z_{11}}&y_{2}{u_{11}}{v_{12}}+y_{1}{z_{12}}&y_{2}{u_{11}}{v_{21}}+y_{1}{z_{21}}&y_{2}{u_{11}}{v_{22}}+y_{1}{z_{22}}\\ y_{2}{u_{12}}{v_{11}}&y_{2}{u_{12}}{v_{12}}&y_{2}{u_{12}}{v_{21}}+{c_{1}}&y_{2}{u_{12}}{v_{22}}\\ y_{2}{u_{21}}{v_{11}}+{c_{2}}&y_{2}{u_{21}}{v_{12}}+{c_{3}}&y_{2}{u_{21}}{v_{21}}&y_{2}{u_{21}}{v_{22}}\\ 
y_{2}{u_{22}}{v_{11}}+y_{1}{z_{11}}&y_{2}{u_{22}}{v_{12}}+y_{1}{z_{12}}&y_{2}{u_{22}}{v_{21}}+{c_{2}}+y_{1}{z_{21}}&y_{2}{u_{22}}{v_{22}}+{c_{3}}+y_{1}{z_{22}}
\end{smatrix}\!.
\end{equation}
This matrix is supposed to be of rank one in~${\BilinearMapCoDomainSubSpace(\Contraction{\tensor{H}}{1}{3},\tensor{P},\HermitianTranspose{\tensor{P}})}$.
Thus, all its~$\matrixsize{2}{2}$ minors are equal to~$0$.
\par
A Gr\"obner basis computation over $\CC$ shows in that case that the relations~${{c_1}^2 c_3=c_1 c_2 c_3=c_1 {c_3}^2=0}$ hold and this is sufficient to conclude.  
Nevertheless, we present a proof that does not require such computations and is valid for any field.
\par
On the one hand, if~$y_{2}=0$, then the constraints of \Cref{eq:type2x} show for instance that both~$\Minor{2,1}{3,3}{\Subs{y_{2}=0}{\mat{G}}}=c_{1}c_{2}$ and~$\Minor{2,2}{3,3}{\Subs{y_{2}=0}{\mat{G}}}=c_{1}c_{3}$ are equal to~$0$.
\par
On the other hand, if~${y_{2}\neq 0}$ then consider the minor~${\Minor{2,3}{3,4}{\mat{G}}}$, which is equal to~$c_{1}y_{2}u_{21}v_{22}$ and supposed to be equal to zero by hypothesis. We are going to explore all the consequences induced by this constraint.
First, if~${y_{1}=0}$, then we have:
\begin{equation}\small
\begin{array}{llll}
u_{21}\neq 0, v_{22}\neq 0 & \multicolumn{3}{l}{\rightarrow c_{1}=0,} \\
u_{21}=0, v_{22}=0 & \multicolumn{3}{l}{\rightarrow \Minor{3,2}{4,4}{\Subs{y_{1}=0,u_{21}=0,v_{22}=0}{\mat{G}}} = {c_{3}}^{2}=0,}\\
u_{21}=0, v_{22}\neq 0&\multicolumn{3}{l}{\rightarrow\Minor{2,2}{3,4}{\Subs{y_{1}=0,u_{21}=0}{\mat{G}}}=-y_{2} u_{12} v_{22} c_3 =0,}\\
u_{21}=0, v_{22}\neq 0,&u_{12}\neq 0& \multicolumn{2}{l}{\rightarrow c_3=0,}\\
u_{21}=0, v_{22}\neq 0,&u_{12}=0 & \multicolumn{2}{l}{\rightarrow \Minor{2,2}{3,3}{\Subs{y_{1}=0,u_{21}=0,u_{12}=0}{\mat{G}}}=-c_1 c_3=0,}\\
u_{21}\neq 0, v_{22}=0 &\multicolumn{3}{l}{\rightarrow \Minor{3,3}{4,4}{\Subs{y_{1}=0,v_{22}=0}{\mat{G}}}= c_3 u_{21} v_{21} y_{2}=0,}\\
u_{21}\neq 0, v_{22}=0, & v_{21}\neq 0 & \multicolumn{2}{l}{\rightarrow c_3=0,}\\
u_{21}\neq 0, v_{22}=0, & v_{21}=0 & \multicolumn{2}{l}{\rightarrow \Minor{2,3}{4,4}{\Subs{y_{1}=0,v_{22}=0,v_{21}=0}{\mat{G}}}= c_1c_3=0,}
\end{array}
\end{equation}
If~$y_{1}$ and~$y_{2}$ are both non-zero, then, the minor $\Minor{1,1}{2,2}{\AddRow{1}{4}{-1}{\mat{G}}}$ is equal to~$y_{2}c_{1}u_{12}v_{12}$ and supposed to be equal to zero by hypothesis.
We are going to explore all the consequences induced by this
constraint. Let $\mat{G}'=\AddRow{1}{4}{-1}{\mat{G}}$:
\begin{equation}\small
\begin{array}{llll}
u_{12}\neq 0 , v_{12}\neq 0 &\multicolumn{3}{l}{\rightarrow c_{1}=0,} \\
u_{12}=0     , v_{12}= 0    &\multicolumn{3}{l}{\rightarrow \Minor{2,2}{3,3}{\Subs{u_{12}=0,v_{12}=0}{\mat{G}'}}= -c_{1}c_{3}=0,}\\
u_{12}=0     , v_{12}\neq 0 &\multicolumn{3}{l}{\rightarrow \Minor{2,3}{3,4}{\Subs{u_{12}=0}{\mat{G}'}}=y_{2}c_{1}u_{21}v_{22}=0,} \\
u_{12}=0     , v_{12}\neq 0, & u_{21}\neq 0, & v_{22}\neq 0 & \rightarrow c_{1}=0,\\
u_{12}=0     , v_{12}\neq 0, & u_{21}=0    & \multicolumn{2}{l}{\rightarrow\Minor{2,2}{3,3}{\Subs{u_{12}=0,u_{21}=0}{\mat{G}'}}=-c_{1}c_{3}=0,}\\
u_{12}=0     , v_{12}\neq 0,& v_{22}=0    & \multicolumn{2}{l}{\rightarrow\Minor{1,3}{2,4}{{\Subs{u_{12}=0,u_{21}=0}{\mat{G}'}}}=c_{1}c_{3}=0,}\\
u_{12}\neq 0 , v_{12}=0     &\multicolumn{3}{l}{\rightarrow\Minor{2,2}{3,4}{\Subs{v_{12}=0}{\mat{G}'}}=-y_{2}u_{12}v_{22}c_{3}=0,}\\
u_{12}\neq 0 , v_{12}=0,    &v_{22}\neq 0 &\multicolumn{2}{l}{\rightarrow c_{3}=0,}\\
u_{12}\neq 0 , v_{12}=0,    &v_{22}=0     &\multicolumn{2}{l}{\rightarrow\Minor{1,2}{3,3}{{\Subs{v_{12}=0,v_{22}=0}{\mat{G}'}}}={c_{3}}^{2}=0.}
\end{array}
\end{equation}
Thus, in any cases, for any such~$\tensor{P}$, the set~${\BilinearMapCoDomainSubSpace(\Contraction{\tensor{H}}{1}{3},\tensor{P},\HermitianTranspose{\tensor{P}})}$ is not~$\BilinearMapCoDomain$.
\end{proof}
The~\cref{prop:deGrooteMainTrick}, together with the computations done in~\cref{fct:2x} and in~\cref{fct:1x} are sufficient to conclude the proof of~\cref{thm:no4product}.
\section{The case of field extensions via matrix polynomial arithmetic}
\label{sec:fieldext}
The cost comparison in \cref{tab:ffcomplex} is for matrices over an arbitrary ring with skew
unitary matrices.
When the ring is an extension, the input of the problem is a polynomial matrix over the base
ring. Following the traditional equivalence between polynomial matrices and matrix polynomials,
leads to alternative ways to multiply the matrix by its transpose, considering the product of two
polynomials with matrix coefficients.
More specifically, we will focus on degree two extensions, and compare the costs in terms of number
of operations over the base ring.

\subsection{The 2M method}\label{ssec:2M}
Over the field~$\CC$ of complex numbers, the~$3M$ method (Karatsuba) for
general matrix multiplication reduces the number of
multiplications of real matrices from~$4$
to~$3$~\cite{Higham:1992:complex3M}:
if~$\MM[\RR]{\omega}(n)$ is the cost of multiplying~${{n}\times{n}}$
matrices over~$\RR$, then the~$3M$~method
costs~$3\MM[\RR]{\omega}(n)+\LO{n^\omega}$ operations over~$\RR$.
Adapting this approach for product of matrix by its adjoint yields a~$2M$ method using only~$2$ real
products:

\begin{algorithm}[htbp]\caption{2M multiplication in an extension}\label{alg:2M}
\begin{algorithmic}
\Require A commutative ring $\K$ and one of its extensions $\E$;
\Require{$\mat{A}\in\K^{m{\times}n}$ and $\mat{B}\in\K^{m{\times}n}$;}
\Require{$\SymbolIAH$ an involutive matrix antihomomorphism of $\E$.}
\Require $i\in\E$, commuting with $\K$, and such that $\epsilon=i\SymbolIAH(i)\in\K$;
\Ensure{${\MatrixProduct{(\mat{A}+i\mat{B})}{\MIAH{A+iB}}}\in\E^{m{\times}m}$.}
\State Let ${\mat{H}=\MatrixProduct{\mat{A}}{\MIAH{B}}}\in\K^{m{\times}m}$;
\State Let ${\mat{G}=\MatrixProduct{(\mat{A}+\mat{B})}{\MIAH{A+\epsilon{B}}}}\in\K^{m{\times}m}$;
\State\Return $(\mat{G}-\epsilon\mat{H}-\MIAH{H})+\mat{H}\SymbolIAH(i)+i\MIAH{H}$.
\end{algorithmic}
\end{algorithm}
\begin{lemma}\label{lem:2M}
\cref{alg:2M} is correct. It costs $2\MM[\K]{\omega}(n)+\LO{n^\omega}$ operations over the base ring
$\K$.
\end{lemma}
\begin{proof} 
  Let
  $\mat{M}={\MatrixProduct{(\mat{A}+i\mat{B})}{\MIAH{A+iB}}}$.
  By~\cref{lem:constant}, we have
  that $\SymbolIAH(i\mat{B})=\MIAH{B}\SymbolIAH(i)=\SymbolIAH(i)\MIAH{B}$.
  Thus, 
  $\mat{M}=\MatrixProduct{\mat{A}}{\MIAH{A}} +
  \MatrixProduct{\mat{A}}{\MIAH{B}}\SymbolIAH(i)+
  i\MatrixProduct{\mat{B}}{\MIAH{A}}+
  i\MatrixProduct{\mat{B}}{\MIAH{B}}\SymbolIAH(i)$, by~\cref{MMIAH:add,MMIAH:mul}.
  As $i$ commutes with $\K$, we also have that 
  $\mat{M}=\MatrixProduct{\mat{A}}{\MIAH{A}}+
  \MatrixProduct{\mat{B}}{\MIAH{B}}\epsilon +
  \MatrixProduct{\mat{A}}{\MIAH{B}}\SymbolIAH(i)+
  i\MatrixProduct{\mat{B}}{\MIAH{A}}$.
  By \cref{MMIAH:invol,MMIAH:mul}, we have that
  $\MIAH{H}=\MatrixProduct{\MIAH{\MIAH{B}}}{\MIAH{A}}=\MatrixProduct{\mat{B}}{\MIAH{A}}$.
  Finally, 
  $\mat{G}=\MatrixProduct{\mat{A}}{\MIAH{A}} +
  \MatrixProduct{\mat{A}}{\MIAH{B}}\epsilon+
  \MatrixProduct{\mat{B}}{\MIAH{A}}+
  \MatrixProduct{\mat{B}}{\MIAH{B}}\epsilon$ as
  $\SymbolIAH(\epsilon)=\SymbolIAH(\SymbolIAH(i))\SymbolIAH(i)=i\SymbolIAH(i)=\epsilon$.
  Therefore
  $\mat{G}-\epsilon\mat{H}-\MIAH{H}=\mat{G}-\mat{H}\epsilon-\MIAH{H}=\MatrixProduct{\mat{A}}{\MIAH{A}}+\MatrixProduct{\mat{B}}{\MIAH{B}}\epsilon$
  and $\mat{M}=\mat{G}+\mat{H}\SymbolIAH(i)+i\MIAH{H}$.
\end{proof}
\begin{example}
For instance, 
if $\K=\RR$, $\E=\CC$, then $i=\sqrt{-1}$ satisfies the conditions of
\cref{alg:2M} for both cases when $\SymbolIAH$ is the
transposition or the conjugate transposition.
Therefore, we obtain the multiplications of a matrix by its adjoint,
whether it be the transpose or the conjugate transpose, in
$2\MM[\RR]{\omega}+\LO{n^{\omega}}$ operations in~$\RR$. 
The classical divide and conquer algorithm, see e.g.~\cite[\S~6.3.1]{jgd:2008:toms}, works directly over~$\CC$ and uses the equivalent of~$\frac{2}{2^{\omega}-4}$ complex floating point~$\matrixsize{n}{n}$ matrix products.
Using the~$3M$ method for the complex products, this algorithm
uses overall~${\frac{6}{2^{\omega}-4}\MM[\RR]{\omega}(n)+\LO{n^\omega}}$ operations in~$\RR$.
Finally, Algorithm~\ref{alg:MIAHMM} 
costs~$\frac{2}{2^{\omega}-3}$ complex multiplications for a leading
term bounded by~$\frac{6}{2^{\omega}-3}\MM[\RR]{\omega}(n)$, improving over~$2\MM[\RR]{\omega}$ for~${\omega>\log_{2}(6)\approx 2.585}$, but
this does not apply to the conjugate transpose case.
This is summarized in Table~\ref{tab:complexcomplex}, also
replacing~$\omega$ by~$3$ or~$\log_{2}(7)$ to illustrate the situation for the
main feasible exponents. 
\begin{table}[htbp]\centering\renewcommand{\arraystretch}{1.5}
\begin{tabular}{lcrrr}
\toprule
Problem & Alg.\ &  $\textrm{MM}_3(n)$ & $\textrm{MM}_{\log_2 7} (n)$ & $\textrm{MM}_\omega(n)$ \\
\midrule
\multirow{2}{*}{$\MatrixProduct{\mat{A}}{\mat{B}}$}
& naive & $8n^{3}$ & $4\,\MM[\RR]{\log_{2}(7)}(n)$& $4\,\MM[\RR]{\omega}(n)$\\
& 3M    & $6n^{3}$ & $3\,\MM[\RR]{\log_{2}(7)}(n)$& $3\,\MM[\RR]{\omega}(n)$\\
\midrule
\multirow{2}{*}{%
  $\MatrixProduct{\mat{A}}{\HermitianTranspose{A}}$}
& Alg.~\ref{alg:2M}   & $4n^{3}$ & \color{darkred}\bf\boldmath$2\,\MM[\RR]{\log_{2}(7)}(n)$& \color{darkred}\bf\boldmath$2\,\MM[\RR]{\omega}(n)$\\
& \cite{jgd:2008:toms} & \color{darkred}\bf\boldmath$3n^{3}$ & \color{darkred}\bf\boldmath$2\,\MM[\RR]{\log_{2}(7)}(n)$
& $\frac{6}{2^{\omega}-4}\,\MM[\RR]{\omega}(n)$\\
\midrule
\multirow{3}{*}{%
  $\MatrixProduct{\mat{A}}{\Transpose{A}}$}
& Alg.~\ref{alg:2M}   & $4n^{3}$ & $2\,\MM[\RR]{\log_{2}(7)}(n)$& \color{darkred}\bf\boldmath$2\,\MM[\RR]{\omega}(n)$\\
& \cite{jgd:2008:toms} & $3n^{3}$ & $2\,\MM[\RR]{\log_{2}(7)}(n)$
& $\frac{6}{2^{\omega}-4}\,\MM[\RR]{\omega}(n)$\\
& Alg.~\ref{alg:MIAHMM}  & \color{darkred}\bf\boldmath$2.4 n^{3}$ & \color{darkred}\bf\boldmath$\frac{3}{2}\,\MM[\RR]{\log_{2}(7)}(n)$ & $\frac{6}{2^{\omega}-3}\,\MM[\RR]{\omega}(n)$\\
\bottomrule
\end{tabular}
\caption{Multiplication of a matrix by its adjoint or transpose over~$\CC$:
  leading term of the cost in number of arithmetic operations over~$\RR$.
  Note that 
  $2<\frac{6}{2^\omega-4}$ only when
  $\omega<\log_2(7)\approx{2.81}$ and
  $2<\frac{6}{2^\omega-3}$ only when
  $\omega<\log_2(6)\approx{2.585}$.}\label{tab:complexcomplex}
\end{table}
\end{example}

\subsection{The quaternion algebra}\label{ssec:quaternions}
Given a field \K of characteristic not $2$, the \K-algebra of \emph{quaternions} \HK is the \K-vector space of all formal linear combinations:
\begin{equation}
x_1+x_2\quat{i}+x_3\quat{j}+x_4\quat{k},\quad (x_1, \ldots, x_4)\in{\K^4},
\end{equation}
the non-commutative multiplication being defined by the bilinear extensions of the
relations:
\begin{equation}
\quat{i}^2=\quat{j}^2=\quat{k}^2=\quat{ijk}=-1.
\end{equation}

The quaternions can also be seen as a degree~$2$ extension of a degree~$2$ extension, but a non-commutative one. 
Therefore the~$3M$ or~§$2M$techniques of~\cref{ssec:2M} only apply directly for the first degree~$2$ extension, while the second extension would require~$4$ multiplications. 
This gives~${{4}\times{3}=12}$ (resp.~${{4}\times{2}=8}$) multiplications in the base field for a general  matrix multiplication (resp.\ a multiplication of a matrix by its transpose or conjugate transpose). 
For the former case, there exist actually algorithms using only~$8$ multiplications instead of~$12$. 
For the latter case, we present algorithms using only~$7$ multiplications for the transpose case and only~$6$ multiplications for the conjugate transpose case instead of~$8$.
\subsubsection{Quaternions' multiplication}
The multiplication of quaternions is
$(x_1+x_2\quat{i}+x_3\quat{j}+x_4\quat{k})(y_1+y_2\quat{i}+y_3\quat{j}+y_4\quat{k})
=(w_1+w_2\quat{i}+w_3\quat{j}+w_4\quat{k})$, with:
\begin{align}
w_1&= x_{1}y_{1}-x_{2}y_{2}-x_{3}y_{3}-x_{4}y_{4}\\
w_2&= x_{1}y_{2}+x_{2}y_{1}+x_{3}y_{4}-x_{4}y_{3}\\
w_3&= x_{1}y_{3}-x_{2} y_{4}+x_{3}y_{1}+x_{4}y_{2}\\
w_4&= x_{1}y_{4}+x_{2}y_{3}-x_{3}y_{2}+x_{4}y_{1}
\end{align}

Fiduccia showed in~\cite{Fiduccia:1971:fmm} how to compute this product with only
$10$ field multiplications and $25$ additions, cleverly using Gau\ss'
trick for the multiplication of complex numbers in three
multiplications.
Regarding the minimal number of base field operation required for the multiplication of quaternions, 
de Groote shows in~\cite{Groote:1975:quaternion} that $10$ multiplications is
minimal to compute both~$\MatrixProduct{\mat{X}}{\mat{Y}}$ and~$\MatrixProduct{\mat{Y}}{\mat{X}}$.
In addition, over the reals and the rationals, the minimal number of
multiplications is~$8$~\cite{Howell:1975:quatprodeight} and
\cite[Proposition~1.7]{groote:1978}.
The algorithm of~\cite{Howell:1975:quatprodeight}, requiring also~$28$
additions, is recalled in Algorithm~\ref{alg:HowellLafon}.

\begin{algorithm}[htbp]\caption{Howell-Lafon quaternion
    multiplication}\label{alg:HowellLafon}
\begin{algorithmic}
\Require $\vec{x}=x_1+x_2\quat{i}+x_3\quat{j}+x_4\quat{k}\in\HK$,
$\vec{y}=y_1+y_2\quat{i}+y_3\quat{j}+y_4\quat{k}\in\HK$
\Ensure $\vec{xy}\in\HK$
\[\begin{array}{rlrl}
Q_{1} &= (x_{1}+x_{2})(y_{1}+y_{2}), &
Q_{2} &= (x_{4}-x_{3})(y_{3}-y_{4})\\
Q_{3} &= (x_{2}-x_{1})(y_{3}+y_{4}), &
Q_{4} &= (x_{3}+x_{4})(y_{2}-y_{1})\\
Q_{5} &= (x_{2}+x_{4})(y_{2}+y_{3}), &
Q_{6} &= (x_{2}-x_{4})(y_{2}-y_{3})\\
Q_{7} &= (x_{1}+x_{3})(y_{1}-y_{4}), &
Q_{8} &= (x_{1}-x_{3})(y_{1}+y_{4})\\
T_1 &= Q_5+Q_6, & T_2 &= Q_7+Q_8\\ T_3 &= Q_5-Q_6, & T_4 &= Q_7-Q_8\\
T_5 &= T_2-T_1, & T_6 &= T_1+T_2\\ T_7 &= T_3+T_4, & T_8 &= T_3-T_4\\
w_1 &= Q_2+T_5/2, & w_2 &= Q_1-T_6/2\\ w_3 &= T_7/2-Q_3, & w_4 &= T_8/2-Q_4
\end{array}\]
\State\Return $w_1+w_2\quat{i}+w_3\quat{j}+w_4\quat{k}$.
\end{algorithmic}
\end{algorithm}

\begin{proposition}\label{prop:HowellLafon}
  \cref{alg:HowellLafon} extends to the case of quaternions with matrix coefficients.
  If  matrix multiplication over the base field costs $\MM[\K]{\omega}(n)$
  field operations for $n{\times}n$ matrices and the field matrix
  addition $O(n^2)$ field additions, then the dominant cost
  of~\cref{alg:HowellLafon} applied to matrices is bounded by
  $8\MM[\K]{\omega}(n)$.
\end{proposition}
\begin{proof}
Correctness is by inspection since $2$ is invertible in $\K$ of
characteristic different from~${2}$.
The complexity bound is just the fact that the Algorithm performs $8$
multiplications of matrices with coefficients in the base field.
\end{proof}

The lowest number of multiplications required to multiply two
quaternions being $8$, \cref{prop:HowellLafon} is the best possible
result while keeping the view of the matrices as two quaternions with
base field matrix coefficients, $X_1,X_2,X_3,X_4$ and
$Y_1,Y_2,Y_3,Y_4$.
The alternative is to use a matrix with quaternion coefficients and
use classical fast matrix algorithms.
Next, we see the different alternatives for the multiplication by an
adjoint.

\subsubsection{Multiplication of a quaternion matrix by its transpose}
We now propose several methods to multiply a quaternion matrix by its
transpose:
\begin{enumerate}
\item First a ``7M'' method which considers a quaternion with matrix
  coefficients, and thus reduces everything to seven general matrix
  multiplications over the base field.
\item Second, one can consider a matrix with quaternion coefficients
  and just apply any matrix multiplication algorithm where
  multiplication of coefficients is that of~\cref{alg:HowellLafon}.
\end{enumerate}

\paragraph{7M method: a quaternion with matrix coefficients}

Many simplifications used in computing the square of a quaternion no longer apply when
computing the product of a quaternion matrix by its transpose, due to non-commutativity of the
matrix product.
For $\mat{A},\mat{B},\mat{C},\mat{D}\in{\K^{m{\times}n}}$,
\begin{multline}
(\mat{A}+\mat{B}\quat{i}+\mat{C}\quat{j}+\mat{D}\quat{k})
(\Transpose{A}+\Transpose{B}\quat{i}+\Transpose{C}\quat{j}+\Transpose{D}\quat{k})
= \mat{S_1}+\mat{S_2}\quat{i}+\mat{S_3}\quat{j}+\mat{S_4}\quat{k}
\end{multline}
where:
\begin{align}
\mat{S_1} &=
\mat{A}\Transpose{A}-\mat{B}\Transpose{B}-\mat{C}\Transpose{C}-\mat{D}\Transpose{D}\\
\mat{S_2} &=
(\mat{A}\Transpose{B}+\mat{B}\Transpose{A})+(\mat{C}\Transpose{D}-\mat{D}\Transpose{C})\\
\mat{S_3} &=
(\mat{A}\Transpose{C}+\mat{C}\Transpose{A})+(\mat{D}\Transpose{B}-\mat{B}\Transpose{D})\\
\mat{S_4} &=
(\mat{A}\Transpose{D}+\mat{D}\Transpose{A})+(\mat{B}\Transpose{C}-\mat{C}\Transpose{B})\end{align}

Using Munro's trick twice, this can be computed with $7$
multiplications over the field \K~and $17$ additions ($6$ of which
are half-additions), as shown in~\cref{alg:QQT}.

\begin{algorithm}[htbp]\caption{Fast quaternion matrix multiplications
  by its transpose}\label{alg:QQT}
\begin{algorithmic}
\Require $\mat{A},\mat{B},\mat{C},\mat{D}\in{\K^{m{\times}n}}$
\Ensure $\MatrixProduct{\mat{M}}{\Transpose{M}}\in\HK^{m{\times}m}$, for $\mat{M}=\mat{A}+\mat{B}\quat{i}+\mat{C}\quat{j}+\mat{D}\quat{k}$.
\[\begin{array}{rlrl}
\mat{U_1} &=\mat{A}+\mat{B}, &\mat{U_2} &=\mat{A}-\mat{B}\\
\mat{U_3} &=\mat{C}+\mat{D}\\
\midrule
\mat{P_1} &=\mat{C}\Transpose{A},
&\mat{P_2} &=\mat{D}\Transpose{B}\\
\mat{P_3} &=\mat{U_3}\Transpose{U_2},
&\mat{P_4} &=\mat{U_1}\Transpose{U_3},\\
\mat{P_5} &=\mat{A}\Transpose{B},
&\mat{P_6} &=\mat{C}\Transpose{D}\\
\mat{P_7} &=(\mat{U_1}+\mat{U_3})(\Transpose{U_2}-\Transpose{U_3})\\
\midrule
\mat{R_1} &= \mat{P_5}+ \mat{P_6}, & \mat{R_2} &= \mat{P_5}-\mat{P_6}\\
\Low{\mat{R_3}} &= \Low{\mat{P_1}+ \Transpose{P_1}} \\
\Low{\mat{R_4}} &= \Low{\mat{P_2}- \Transpose{P_2}}\\
\Low{\mat{S_1}} &= \Low{\mat{P_7} -\mat{P_3}+\mat{P_4}+\mat{R_1}-\Transpose{R_2}}\\
\mat{S_2} &= \mat{R_1} +\Transpose{R_2}\\
\mat{S_3} &= \mat{R_3} +\mat{R_4}\\
\mat{S_4} &= \mat{P_3}+\mat{P_4}-\Transpose{S_3}
\end{array}\]
\State\Return \mat{S_1}+\mat{S_2}\quat{i}+\mat{S_3}\quat{j}+\mat{S_4}\quat{k}
\end{algorithmic}
\end{algorithm}

\begin{problem} Multiply a quaternion with matrix coefficients by its
  transpose in fewer than $7$ multiplications. \end{problem}

\paragraph{Using matrices of quaternions and divide and conquer}
In the following, for the sake of simplicity, we will consider only
square matrices.
Here, we consider instead a matrix with quaternion coefficients and
perform a matrix-matrix product: since the quaternions are not
commutative, then $\mat{M}\Transpose{M}$ is not necessarily symmetric,
one has to compute both the top right and bottom left corners of the
product. Thus no gain is obvious between computing
$\mat{M}\Transpose{M}$ and $\mat{M}\mat{N}$ that way
and~\cref{alg:MIAHMM} is a priori useless in this case, as remarked in
the (counter)-\cref{ex:antih}.
The idea is thus to use a non symmetric algorithm, applied to
$\mat{M}$ and $\Transpose{M}$. The baseline cost would then again be
$\MM[\HK]{\omega}(n)=8\MM[\K]{\omega}(n)$.

Another approach is to use a divide and conquer strategy at the higher
level: cut $\mat{M}$ into
$\begin{bmatrix}\mat{M_{11}}&\mat{M_{12}}\\\mat{M_{21}}&\mat{M_{22}}\end{bmatrix}$, and compute:
\begin{itemize}
\item
$\MatrixProduct{\mat{M_{11}}}{\Transpose{M_{21}}}$,
$\MatrixProduct{\mat{M_{12}}}{\Transpose{M_{22}}}$,
$\MatrixProduct{\mat{M_{21}}}{\Transpose{M_{11}}}$,
$\MatrixProduct{\mat{M_{22}}}{\Transpose{M_{12}}}$
by the baseline algorithm;
\item and
$\MatrixProduct{\mat{M_{11}}}{\Transpose{M_{11}}}$,
$\MatrixProduct{\mat{M_{12}}}{\Transpose{M_{12}}}$,
$\MatrixProduct{\mat{M_{21}}}{\Transpose{M_{21}}}$,
$\MatrixProduct{\mat{M_{22}}}{\Transpose{M_{22}}}$
by recursive calls.
\end{itemize}
The cost of this divide and conquer strategy is then:
\begin{equation}
C(n)\leq{}4C(\frac{n}{2})+4\MM[\HK]{\omega}(\frac{n}{2})+\LO{n^\omega}\leq{}4C(\frac{n}{2})+32\MM[\K]{\omega}(\frac{n}{2})+\LO{n^\omega}.\end{equation}
By~\cref{lem:masterthm}, we have that
$C(n)\leq{\frac{32}{2^\omega-4}\MM[\K]{\omega}(n)+\LO{n^\omega}}$, and
this is
never better that $8\MM[\K]{\omega}(n)$ (but equal when~$\omega=3$ as
expected). So this is thus useless too.

But the same strategy can be used with a Strassen-like algorithm
instead. Now such algorithms, for instance those of
\cite{Strassen:1969:GENO,Winograd:1977:complexite}, when applied to~$\mat{M}$ and~$\Transpose{M}$, use two recursive calls and five normal
multiplications.
This is:
\begin{equation}
S(n)\leq{2S(\frac{n}{2})+5\MM[\HK]{\omega}(\frac{n}{2})+\LO{n^\omega}}\leq{2S(\frac{n}{2})+40\MM[\K]{\omega}(\frac{n}{2})+\LO{n^\omega}}.
\end{equation}
By~\cref{lem:masterthm}, we obtain that this is
\begin{equation}\label{eq:hybrid}
S(n)\leq{\frac{40}{2^\omega-2}\MM[\K]{\omega}(n)+\LO{n^\omega}}.
\end{equation}
As expected this is again $8\MM[\K]{\log_2(7)}(n)$ if a Strassen-like
algorithm is also used for the baseline over the field and
$\omega=\log_2(7)$.
This is worse if $\omega<\log_2(7)$, but better, and only
$(6+\frac{2}{3})\MM[\K]{3}(n)$, if $\omega=3$.

\paragraph{Positive characteristic quaternions and transposition}
In this case, \cref{alg:MIAHMM} is not usable. It nonetheless has an
interesting feature: it has $3$ symmetric products
instead of $2$ for the algorithms of
\cite{Strassen:1969:GENO,Winograd:1977:complexite}.
In the quaternion case, as transposition is not an antihomomorphism,
one cannot use the symmetries directly to save computations as in
general
$\Transpose{\left(\MatrixProduct{M}{N}\right)}\neq\MatrixProduct{\Transpose{N}}{\Transpose{M}}$.
But if one is willing to \emph{recompute}
$\MatrixProduct{\Transpose{N}}{\Transpose{M}}$
then the algorithm still works.
This yields to~\cref{alg:nonAIH} which requires $7$ multiplications
instead of $5$, $15$ additions instead of $9$, and $4$ multiplications
by $\mat{Y}$ or $\Transpose{Y}$.

\begin{algorithm}[htb]
\caption{Product of a matrix by its non antihomomorphic
  transpose}\label{alg:nonAIH}
\begin{algorithmic}
\Require{$\mat{A}\in \Matrices{\Ring}{\matrixsize{m}{n}}$ (with even $m$ and $n$ for the
  sake of simplicity);}
\Require $\mat{Y} \in \Matrices{\Ring}{\matrixsize{\frac{n}{2}}{\frac{n}{2}}}$ such that
$\MatrixProduct{\mat{Y}}{\Transpose{Y}}=-\IdentityMatrixOfSize{\frac{n}{2}}$;
\Ensure{$\MatrixProduct{\mat{A}}{\Transpose{{A}}}$.}
\State Split
$\mat{A}=\begin{smatrix} \mat{A}_{11}&\mat{A}_{12}\\
  \mat{A}_{21}&\mat{A}_{22}\end{smatrix}$
where~$\mat{A}_{11}$ is in~${\Matrices{\Ring}{\matrixsize{\frac{m}{2}}{\frac{n}{2}}}}$
\[\begin{array}{lcl}
{\mat{S}_{1}}\leftarrow{\MatrixProduct{(\mat{A}_{21} -
    \mat{A}_{11})}{\mat{Y}}}
&&
{\mat{ST}_{1}}\leftarrow{\MatrixProduct{\Transpose{Y}}{(\Transpose{A_{21}} - \Transpose{A_{11}})}}
\\
{\mat{S}_{2}}\leftarrow{\mat{A}_{22} - \MatrixProduct{\mat{A}_{21}}{\mat{Y}}}
&&
{\mat{ST}_{2}}\leftarrow{\Transpose{A_{22}} - \MatrixProduct{\Transpose{Y}}{\Transpose{A_{21}}}}
\\
{\mat{S}_{3}}\leftarrow{\mat{S}_{1} - \mat{A}_{22}}
&&
{\mat{ST}_{3}}\leftarrow{\mat{ST}_{1} - \Transpose{A_{22}}}
\\
{\mat{S}_{4}}\leftarrow{\mat{S}_{3} + \mat{A}_{12}}
&&
{\mat{ST}_{4}}\leftarrow{\mat{ST}_{3} + \Transpose{A_{12}}}
\\
\midrule
\multicolumn{3}{c}{\mat{{P}_{1}}\leftarrow{\MatrixProduct{\mat{A}_{11}}{\Transpose{{A}_{11}}}}}
\\
\multicolumn{3}{c}{\mat{{P}_{2}}\leftarrow{\MatrixProduct{\mat{A}_{12} }{\Transpose{{A}_{12}}}}}
\\
{\mat{P}_{3} \leftarrow\MatrixProduct{\mat{A}_{22}}{\mat{{ST}_{4}}}}
&&
{\mat{PT}_{3} \leftarrow\MatrixProduct{\mat{{S}_{4}}}{\Transpose{A_{22}}}}
\\
{\mat{P}_{4} \leftarrow\MatrixProduct{\mat{S}_{1}}{\mat{{ST}_{2}}}}
&&
{\mat{PT}_{4} \leftarrow\MatrixProduct{\mat{{S}_{2}}}{\mat{ST}_{1}}}
\\
\multicolumn{3}{c}{\mat{{P}_{5}}\leftarrow{\MatrixProduct{\mat{S}_{3}}{\mat{{ST}_{3}}}}}
\\
\midrule
\multicolumn{3}{c}{\mat{U}_{1}\leftarrow\mat{{P}_{1}}+\mat{{P}_{5}}}
\\
\multicolumn{3}{c}{\mat{U}_{3} \leftarrow \mat{P}_{1} + \mat{P}_{2}}
\\
{\mat{U}_{2} \leftarrow \mat{U}_{1} + \mat{P}_{4}}
&&
{\mat{UT}_{2} \leftarrow \mat{U}_{1} + \mat{PT}_{4}}
\\
\mat{U}_{4} \leftarrow \mat{U}_{2} + \mat{P}_{3}
&&
\mat{UT}_{4} \leftarrow \mat{UT}_{2} + \mat{PT}_{3}
\\
\multicolumn{3}{c}{\mat{U}_{5} \leftarrow \mat{U}_{2} + \mat{PT}_{4}}
\end{array}\]
\State \Return{$\begin{smatrix} \mat{U}_{3} &\mat{UT}_{4} \\ \mat{U}_{4} & \mat{U}_{5} \end{smatrix}$.}
\end{algorithmic}
\end{algorithm}

Now, it turns out that transposition is still antihomomorphic if one
of the matrices has its coefficients in the base field, as shown
by~\cref{lem:basetrsp}.
\begin{lemma}\label{lem:basetrsp}
 Let~$\mat{A}$ be in~$\HK[\K]^{m{\times}k}$ and~$\mat{Y}$ in~$\K^{k{\times}n}$,
  then
  $\Transpose{\left(\MatrixProduct{\mat{A}}{\mat{Y}}\right)}=\MatrixProduct{\Transpose{Y}}{\Transpose{A}}$.
\end{lemma}
\begin{proof}
Since the coefficients of $\mat{Y}$ are in the base field, they commute
with the quaternions. Therefore, $\forall{i,j}$, $\sum_k a_{jk}y_{ki}=\sum_k y_{ki}a_{jk}$.
\end{proof}

Now, \cref{ssec:tridiag} shows that for any 
quaternion algebra in positive characteristic, there exist a matrix $\mat{Y}$, \emph{in the base field},
such that
$\MatrixProduct{\mat{Y}}{\Transpose{Y}}=-\IdentityMatrixOfSize{\lfloor
  \frac{n}{2}\rfloor}$.
Therefore, in this case, \cref{lem:basetrsp} shows that
in~\cref{alg:nonAIH},
$\mat{ST}_1=\Transpose{S_1}$,
$\mat{ST}_2=\Transpose{S_2}$,
$\mat{ST}_3=\Transpose{S_3}$,
$\mat{ST}_4=\Transpose{S_4}$.
This shows that not only $\mat{P}_1$ and $\mat{P}_2$ are
multiplications of a matrix by its transpose, but also
$\mat{P_5}=\MatrixProduct{\mat{S_3}}{\Transpose{S_3}}$. Finally,
\cref{alg:nonAIH} thus requires three recursive calls and four general
multiplications.
This is:
\begin{equation}
P(n)\leq{3P(\frac{n}{2})+4\MM[\HK]{\omega}(\frac{n}{2})+\LO{n^\omega}}\leq{3P(\frac{n}{2})+32\MM[\K]{\omega}(\frac{n}{2})+\LO{n^\omega}}\end{equation}
and \cref{eq:hybrid} is modified as:
\begin{equation}\label{eq:primehybrid}
P(n)\leq{\frac{32}{2^\omega-3}\MM[\K]{\omega}(n)+\LO{n^\omega}}
\end{equation}
As expected this is again $8\MM[\K]{\log_2(7)}(n)$ if a Strassen-like
algorithm is also used for the baseline over the field and
$\omega=\log_2(7)$.
This is again worse if $\omega<\log_2(7)$, but better, and only
$(6+\frac{2}{5})\MM[\K]{3}(n)$, if $\omega=3$.

\newcommand{\OurAlg}{\cref{alg:MIAHMM}}

\subsubsection{Multiplication of a quaternion matrix by its adjoint}

We now deal with the case of the product of a quaternion matrix with its conjugate transpose.
This operator is now an antihomomorphism which allows us to save some computations as in~\OurAlg{}.
Here also we distinguish the matrix of quaternions from the quaternion with matrix
coefficients.

\paragraph{Scalar case}
The quaternion conjugation satisfies
$\overline{XY}=\overline{Y}\,\overline{X}$. Therefore, we have:
\begin{equation}\label{eq:quatnorm}
(a+b\quat{i}+c\quat{j}+d\quat{k})
\overline{(a+b\quat{i}+c\quat{j}+d\quat{k})} =
a^{2}+b^{2}+c^{2}+d^{2}
\end{equation}
For matrices again simplifications do not occur and the product is
then more complex.

\paragraph{Using matrices of quaternions}
Now $\MatrixProduct{\mat{M}}{\HermitianTranspose{M}}$ is a hermitian matrix
and~\OurAlg{} works over \HK. For this, one needs to find a
skew-unitary matrix in \HK. This is impossible in \HK[\CC], but always
possible in the quaternions over fields of positive characteristic using sums
of squares and~\cref{eq:quatnorm}.

Suppose we use a generic matrix multiplication
algorithm over the quaternions with cost bound equivalent to
$\MM[\HK]{\omega}(n)$ field operations.
Then our~\cref{alg:MIAHMM} can multiply a matrix of a quaternions by its
conjugate transpose with a dominant complexity term bounded by
$\left(\frac{2}{2^\omega-3}\right)\MM[\HK]{\omega}(n)$
operations, by~\cref{thm:complexitybound}.

Now for the quaternions, the best algorithm to multiply any two
matrices of quaternions is given by~\cref{prop:HowellLafon} and uses
$8\MM[\K]{\omega}(n)$ field operations if the base field matrix
multiplication uses $\MM[\K]{\omega}(n)$.
We thus have proven:
\begin{corollary}\label{cor:algquat}
\OurAlg{} multiplies a quaternion matrix by its
conjugate transpose with dominant cost bounded by
$\left(\frac{16}{2^\omega-3}\right)\MM[\K]{\omega}(n)$ base field operations.
\end{corollary}

\paragraph{Directly using quaternions with matrix coefficients}
Using a quaternion with matrix coefficients over
the field, we have:
\begin{equation}\begin{split}
\MatrixProduct{\mat{M}}{\HermitianTranspose{M}} &=
\MatrixProduct{(\mat{A}+\mat{B}\quat{i}+\mat{C}\quat{j}+\mat{D}\quat{k})}%
{\HermitianTranspose{(\mat{A}+\mat{B}\quat{i}+\mat{C}\quat{j}+\mat{D}\quat{k})}}\\
&=(\mat{A}+\mat{B}\quat{i}+\mat{C}\quat{j}+\mat{D}\quat{k})
(\Transpose{A}-\Transpose{B}\quat{i}-\Transpose{C}\quat{j}-\Transpose{D}\quat{k})\\
&= \mat{H_1}+\mat{H_2}\quat{i}+\mat{H_3}\quat{j}+\mat{H_4}\quat{k}
\end{split}\end{equation}
where:
\begin{align}
\mat{H_1} &=
\mat{A}\Transpose{A}+\mat{B}\Transpose{B}+\mat{C}\Transpose{C}+\mat{D}\Transpose{D}\\
\mat{H_2} &=
(\mat{B}\Transpose{A}-\mat{A}\Transpose{B})+(\mat{D}\Transpose{C}-\mat{C}\Transpose{D})\\
\mat{H_3} &=
(\mat{C}\Transpose{A}-\mat{A}\Transpose{C})+(\mat{B}\Transpose{D}-\mat{D}\Transpose{B})\\
\mat{H_4} &=
(\mat{D}\Transpose{A}-\mat{A}\Transpose{D})+(\mat{C}\Transpose{B}-\mat{B}\Transpose{C})\end{align}

Note that $H_1$ is symmetric and $H_2$, $H_3$, $H_4$ are skew-symmetric.

The properties of the transpose in the field shows that these can be
computed with $4+3*2=10$ multiplications (including four squares).

Now consider
$\mat{E}=\mat{A}+\mat{B}\quat{i}$ and $\mat{F}=\mat{C}+\mat{D}\quat{i}$,
so that
$\mat{M}=\mat{E}+\mat{F}\quat{j}$.
This shows that:
\begin{equation}
\MatrixProduct{\mat{M}}{\HermitianTranspose{M}}
=\MatrixProduct{(\mat{E}+\mat{F}\quat{j})}{(\HermitianTranspose{E}+\quat{j}\HermitianTranspose{F})}
=
(\mat{E}\HermitianTranspose{E}
+
\mat{F}\HermitianTranspose{F})
+
(\mat{F}\quat{j}\HermitianTranspose{E}
-
\mat{E}\quat{j}\HermitianTranspose{F})
\end{equation}

Then, we have:
\begin{equation}\label{eq:quaternion3Ma}
\begin{split}
\mat{F}\quat{j}\HermitianTranspose{E} &=
(\mat{C}+\mat{D}\quat{i})\quat{j}(\Transpose{A}-\Transpose{B}\quat{i})\\
&= (\mat{C}\Transpose{A}-\mat{D}\Transpose{B})\quat{j} +
(\mat{D}\Transpose{A}+\mat{C}\Transpose{B})\quat{k}\\
& =\mat{X}\quat{j}+\mat{Y}\quat{k}\\
\end{split}
\end{equation}
\begin{equation}\label{eq:quaternion3Mb}
\begin{split}
-\mat{E}\quat{j}\HermitianTranspose{F} &=
(\mat{A}+\mat{B}\quat{i})\quat{j}(-\Transpose{C}+\Transpose{D}\quat{i})\\
&= (-\mat{A}\Transpose{C}+\mat{B}\Transpose{D})\quat{j} -
(\mat{B}\Transpose{C}+\mat{A}\Transpose{D})\quat{k}\\
& =-\Transpose{X}\quat{j}-\Transpose{Y}\quat{k}
\end{split}
\end{equation}

Using \cref{eq:quaternion3Ma,eq:quaternion3Mb}, we thus have \cref{alg:2M3M} which
uses only $6$ multiplications (one of which is a square) and a
total of $14$ additions, $7$ of them being half-additions
(On the one hand, $3$ multiplications and $5$ additions for
\cref{eq:quaternion3Ma,eq:quaternion3Mb} overall, then $2$
half-additions for $\Low{\mat{H_3}}$ and $\Low{\mat{H_4}}$;
on the other hand, $2$ multiplications, $1$ square, $3$ additions and $9$
half-additions for $\Low{\mat{H_1}}$ and $\Low{\mat{H_2}}$).

\begin{algorithm}[htbp]\caption{Fast quaternion matrix multiplications
  by its adjoint}\label{alg:2M3M}
\begin{algorithmic}
\Require $\mat{A},\mat{B},\mat{C},\mat{D}\in{\K^{m{\times}n}}$
\Ensure $\MatrixProduct{\mat{M}}{\HermitianTranspose{M}}\in\HK^{m{\times}m}$, for $\mat{M}=\mat{A}+\mat{B}\quat{i}+\mat{C}\quat{j}+\mat{D}\quat{k}$.
\[\begin{array}{rlrl}
\mat{U_1} &= \mat{A}+\mat{B}, &\mat{U_2} &= \mat{C}+\mat{D}\\
\midrule
\mat{Q_1} &=\mat{C}\Transpose{A},
&\mat{Q_2} &=\mat{D}\Transpose{B}\\
\mat{Q_3} &=\mat{U_2}\Transpose{U_1}\\
 \mat{Q_4} &=\mat{A}\Transpose{B},
&\mat{Q_5} &=\mat{C}\Transpose{D}\\
\mat{Q_6} &=(\mat{U_1}+\mat{U_2})(\Transpose{U_1}+\Transpose{U_2})\\
\midrule
\mat{T_1} &=  \mat{Q_1}-\mat{Q_2},
&\mat{T_2} &=  \mat{Q_4}+\mat{Q_5}\\
\mat{T_3} &=  (\mat{Q_1}+\mat{Q_2})-\mat{Q_3}\\
\Low{\mat{H_1}} &= \Low{\mat{Q_6}-(\Transpose{Q_3}+\mat{Q_3})-(\Transpose{T_2}+\mat{T_2})}\\
\Low{\mat{H_2}} &=  \Low{\Transpose{T_2}-\mat{T_2}}\\
\Low{\mat{H_3}} &= \Low{\mat{T_1}-\Transpose{T_1}}\\
\Low{\mat{H_4}} &= \Low{\Transpose{T_3}-\mat{T_3}}\\
\end{array}\]
\State\Return $\mat{H_1}+\mat{H_2}\quat{i}+\mat{H_3}\quat{j}+\mat{H_4}\quat{k}$.
\end{algorithmic}
\end{algorithm}

\begin{proposition}\label{prop:6M}
\cref{alg:2M3M} multiplies a quaternion matrix by its
conjugate transpose with cost equivalent
to~$\left(5+\frac{2}{2^\omega-3}\right)\MM[\K]{\omega}(n)$ base field operations.
\end{proposition}
\begin{proof}
\cref{alg:2M3M} uses $6$ multiplications, one of which,
$\mat{Q_6}$, is the product of a matrix  in~$\K$ by its transpose.
\end{proof}

\begin{problem} Multiply a quaternion with matrix coefficients by its
  Hermitian transpose in fewer than $6$ multiplications (including
  $1$ square), or with more squares. \end{problem}

\paragraph{Comparison}
We summarize the results of this section about quaternion matrices
in~\cref{tab:complexquat}.

\begin{table}[htbp]\small\centering\renewcommand{\arraystretch}{1.5}
\begin{tabular}{llrrr}
\toprule
& Alg.\ &  $\MM{3}(n)$ & $\MM{\log_2 7} (n)$ & $\MM{\omega}(n)$ \\
\midrule
\multirow{2}{*}{$\MatrixProduct{\mat{A}}{\mat{B}}$}
& naive & $32n^3$ & $16\,\MM[\K]{\log_{2}(7)}(n)$& $16\,\MM[\K]{\omega}(n)$\\
& \cite{Howell:1975:quatprodeight}  & $16n^3$ & $8\,\MM[\K]{\log_{2}(7)}(n)$& $8\,\MM[\K]{\omega}(n)$\\
\midrule
\multirow{2}{*}{$\MatrixProduct{\mat{A}}{\HermitianTranspose{A}}$}
& Alg.~\ref{alg:2M3M}    & $10.8n^{3}$ & $5.5\,\MM[\K]{\log_{2}(7)}(n)$& \color{darkred}\bf\boldmath$\left(\frac{2}{2^\omega-3}+5\right)\MM[\K]{\omega}(n)$\\
& Alg.~\ref{alg:MIAHMM}  & \color{darkred}\bf\boldmath$6.4 n^{3}$ & \color{darkred}\bf\boldmath$4\,\MM[\K]{\log_{2}(7)}(n)$ & $\frac{16}{2^{\omega}-3}\,\MM[\K]{\omega}(n)$\\
\midrule
\multirow{2}{*}{$\MatrixProduct{\mat{A}}{\Transpose{A}}$}
& Alg.~\ref{alg:QQT}    & $14n^{3}$ & \color{darkred}\bf\boldmath$7\,\MM[\K]{\log_{2}(7)}(n)$& \color{darkred}\bf\boldmath$7\,\MM[\K]{\omega}(n)$\\
& Eq.~(\ref{eq:hybrid}) & \color{darkred}\bf\boldmath$(13+\frac{1}{3})n^{3}$ & $8\,\MM[\K]{\log_{2}(7)}(n)$& $\frac{40}{2^\omega-2}\MM[\K]{\omega}(n)$\\
\hdashline
$>0$ char. & Alg.~\ref{alg:nonAIH} & \color{darkred}\bf\boldmath$(12+\frac{4}{5})n^{3}$ & $8\,\MM[\K]{\log_{2}(7)}(n)$& $\frac{32}{2^\omega-3}\MM[\K]{\omega}(n)$\\
\bottomrule
\end{tabular}
\caption{Matrix Multiplication over~$\HK^{n{\times}n}$: leading term of the cost
 in number of operations over~$\K$. Note that
 $\left(\frac{2}{2^\omega-3}+5\right)<\frac{16}{2^\omega-3}$ only when
 $\omega<{\log_2(29)-\log_2(5)}\approx{2.536}$.}\label{tab:complexquat}
\end{table}
\section{Algorithm into practice}\label{sec:implem}
This section reports on an implementation of \cref{alg:MIAHMM} over a prime field, as it is
a core ingredient of any such computation in positive characteristic or over $\Z[i]$ or $\QQ[i]$.
In order to reduce the memory footprint and increase the data locality of the computation, we first
need to identify a memory placement 
and a scheduling of the tasks minimizing the temporary allocations.
We thus propose in Table~\ref{tab:schedule:AAT} and
Figure~\ref{fig:DAG:AAT} a memory placement and schedule for the
operation~${C\leftarrow\MatrixProduct{\mat{A}}{\Transpose{A}}}$ using
no more extra storage than the unused upper triangular part of the
result~$C$.

\begin{table}[htbp]
  \begin{center}
    \begin{tabular}{cll|cll}
      \toprule
      \# & operation & loc.\ & \# & operation & loc.\\
      \midrule
	1&$S_{1}=\MatrixProduct{(A_{21}-A_{11})}{Y}$	&$C_{21}$&9&$U_{1}=P_{1}+P_{5}$&$C_{12}$\\
	2&$S_{2}=A_{22}-\MatrixProduct{A_{21}}{Y}$	&$C_{12}$&&$\text{Up}(U_{1})=\Transpose{\text{Low}(U_{1})}$&$C_{12}$\\
	3&$\Transpose{P_{4}}=\MatrixProduct{S_{2}}{\Transpose{S_{1}}}$&$C_{22}$&10&$U_{2}=U_{1}+P_{4}$&$C_{12}$\\
	4&$S_{3}=S_{1}-A_{22}$&$C_{21}$&11&$U_{4}=U_{2}+P_{3}$&$C_{21}$\\
	5&$P_{5}=\MatrixProduct{S_{3}}{\Transpose{S_{3}}}$&$C_{12}$&12&$U_{5}=U_{2}+\Transpose{P_{4}}$&$C_{22}$\\
	6&$S_{4}=S_{3}+A_{12}$&$C_{11}$&13&$P_{2}=\MatrixProduct{A_{12}}{\Transpose{A_{12}}}$&$C_{12}$\\
	7&$P_{3}=\MatrixProduct{A_{22}}{\Transpose{S_{4}}}$&$C_{21}$&14&$U_{3}=P_{1}+P_{2}$&$C_{11}$\\
	8&$P_{1}=\MatrixProduct{A_{11}}{\Transpose{A_{11}}}$&$C_{11}$\\
      \bottomrule
    \end{tabular}
    \caption{%
Memory placement and schedule of tasks to compute the lower triangular part of~${C\leftarrow \MatrixProduct{A}{\Transpose{A}}}$ when~${k\leq n}$.
The block~$C_{12}$ of the output matrix is the only temporary used.
}\label{tab:schedule:AAT}
  \end{center}
\end{table}

\begin{figure}[htbp]
  \begin{center}
\begin{tikzpicture}%
  \matrix (m) [matrix of math nodes, row sep=6pt, column sep=50pt ]
  {%
    C_{22} & C_{12} & C_{21} & C_{11} \\
          & S_{2}   & S_{1}   &       \\
     \Transpose{P_{4}}  &       & S_{3}   &       \\
          & P_{5}   &       &  S_{4}  \\
          &       & P_{3}   &       \\
          &       &       &  P_{1}  \\
          &  U_{1}  &       &       \\
          &  U_{2}  &       &       \\
   U_{5}    &       & U_{4}   &       \\
          &  P_{2}  &       &       \\
          &       &       &  U_{3}  \\
  };
  \path[-stealth]
  (m-2-2) edge (m-3-1)
  (m-2-3) edge (m-3-1)
          edge (m-3-3)
  (m-3-3) edge (m-4-2)
          edge (m-4-4)
  (m-4-4) edge (m-5-3)
  (m-4-2) edge (m-7-2)
  (m-6-4) edge (m-7-2)
  (m-7-2) edge (m-8-2)
  (m-3-1) edge (m-8-2)
          edge (m-9-1)
  (m-8-2) edge (m-9-1)
          edge (m-9-3)
  (m-5-3) edge (m-9-3)
  (m-10-2) edge (m-11-4)
  (m-6-4) edge (m-11-4);
\end{tikzpicture}
\caption{\textsc{dag} of the tasks and their memory location for the computation of~${C\leftarrow \MatrixProduct{A}{\Transpose{A}}}$ presented in Table~\ref{tab:schedule:AAT}.}\label{fig:DAG:AAT}
\end{center}
\end{figure}

The more general operation~${C\leftarrow \alpha \MatrixProduct{A}{\Transpose{A}} + \beta C}$,
is referred to as \texttt{SYRK} (Symmetric Rank $k$ update) in the \textsc{blas} \textsc{api}.
Table~\ref{tab:schedule:AATpC} and Figure~\ref{fig:DAG:AATpC} propose
a schedule requiring only one additional~${{n/2}\times{n/2}}$ temporary
storage.

\begin{table}[htbp]
\begin{center}
\begin{tabular}{ll|ll}
\toprule
 operation & loc.\ & operation & loc.\\
\midrule
${S_{1}=\MatrixProduct{(A_{21}-A_{11})}{Y}}$&tmp&$P_{1}=\alpha\MatrixProduct{A_{11}}{\Transpose{A_{11}}}$&tmp\\
${S_{2}=A_{22}-\MatrixProduct{A_{21}}{Y}}$&$C_{12}$&$U_{1}=P_{1}+P_{5}$&$C_{12}$\\
$\text{Up}(C_{11})=\Transpose{\text{Low}(C_{22})}$&$C_{11}$&$\text{Up}(U_{1})=\Transpose{\text{Low}(U_{1})}$&$C_{12}$\\
$\Transpose{P_{4}}=\alpha\MatrixProduct{S_{2}}{\Transpose{S_{1}}}$&$C_{22}$&$U_{2}=U_{1}+P_{4}$&$C_{12}$\\
$S_{3}=S_{1}-A_{22}$&tmp&$U_{4}=U_{2}+P_{3}$&$C_{21}$\\
$P_{5}=\alpha\MatrixProduct{S_{3}}{\Transpose{S_{3}}}$&$C_{12}$&$U_{5}=U_{2}+\Transpose{P_{4}}+\beta\Transpose{\text{Up}(C_{11})}$&$C_{22}$\\
$S_{4}=S_{3}+A_{12}$&tmp&$P_{2}=\alpha\MatrixProduct{A_{12}}{\Transpose{A_{12}}}+\beta C_{11}$&$C_{11}$\\
$P_{3}=\alpha\MatrixProduct{A_{22}}{\Transpose{S_{4}}}+\beta C_{21}$&$C_{21}$&$U_{3}=P_{1}+P_{2}$&$C_{11}$\\
\bottomrule
\end{tabular}
\caption{%
Memory placement and schedule of tasks to compute the lower triangular part of~${C\leftarrow \alpha \MatrixProduct{A}{\Transpose{A}}+\beta C}$ when~${k\leq n}$.
The block~$C_{12}$ of the output matrix as well as an~${n/2\times n/2}$ block tmp are used as temporary storage.}\label{tab:schedule:AATpC}
\end{center}
\end{table}

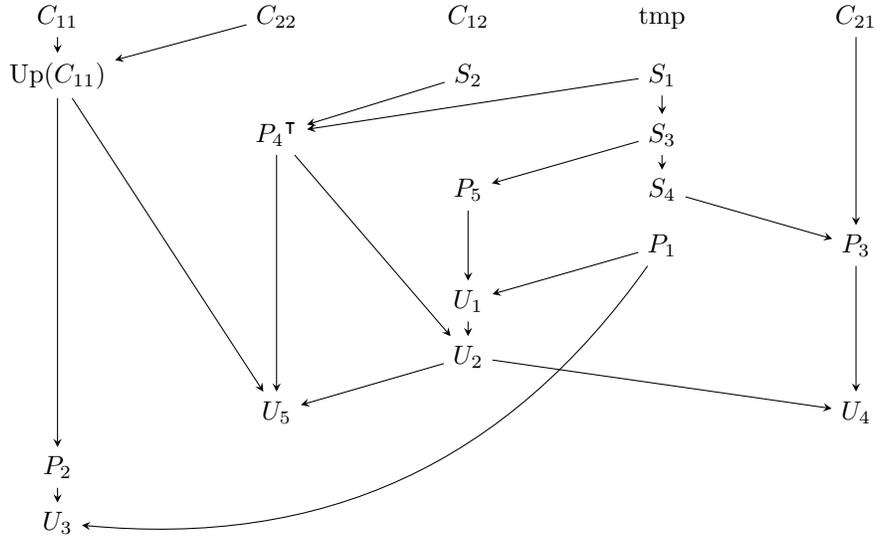
\begin{figure}[htbp]
  \begin{center}
\begin{tikzpicture}%
  \matrix (m) [matrix of math nodes, row sep=6pt, column sep=50pt ]
  {%
   C_{11}  &C_{22} & C_{12} & \text{tmp} &C_{21} \\
\text{Up}(C_{11})& & S_{2}   & S_{1}   &        \\
          & \Transpose{P_{4}}&  & S_{3}   &    \\
          &       & P_{5}   & S_{4}   &       \\
          &       &       & P_{1}   &  P_{3}  \\
          &       &  U_{1}  &       &       \\
          &       &  U_{2}  &       &       \\
          &  U_{5}  &       &       &  U_{4}    \\
    P_{2}   &       &       &       &        \\
    U_{3}   &       &       &       &       \\
  };
  \path[-stealth]
  (m-1-2) edge (m-2-1)
  (m-1-1) edge (m-2-1)
  (m-2-3) edge (m-3-2)
  (m-2-4) edge (m-3-2)
          edge (m-3-4)
  (m-3-4) edge (m-4-3)
          edge (m-4-4)
  (m-4-4) edge (m-5-5) %
  (m-4-3) edge (m-6-3) %
  (m-6-3) edge (m-7-3) %
  (m-3-2) edge (m-7-3) %
          edge (m-8-2) %
  (m-5-4) edge (m-6-3) %
  (m-5-5) edge (m-8-5) %
  (m-7-3) edge (m-8-5)  %
          edge (m-8-2) %
  (m-5-4) edge[bend left] (m-10-1) %
  (m-9-1) edge (m-10-1) %
  (m-1-5) edge (m-5-5) %
  (m-2-1) edge (m-8-2) %
          edge (m-9-1); %
\end{tikzpicture}
\caption{\textsc{dag} of the tasks and their memory location for the computation of~${C\leftarrow \alpha \MatrixProduct{A}{\Transpose{A}} + \beta C}$ presented in Table~\ref{tab:schedule:AATpC}.}\label{fig:DAG:AATpC}
  \end{center}
  \end{figure}

These algorithms have been implemented as the \texttt{fsyrk} routine
in the open source \texttt{fflas-ffpack} library for dense linear
algebra over a finite
field~\cite[\href{https://github.com/linbox-team/fflas-ffpack/commit/0a91d61e6518568b006873076df925fcd6fcc112}{from
  commit 0a91d61e}]{fflas19}.

Figure~\ref{fig:perfs} compares the computation speed in effective
Gfops (a normalization, defined as~${n^{3}/(10^{9}\times\textrm{time})}$) of this
implementation over~${\Z/131071\Z}$ with that of the double precision
\textsc{blas} routines \texttt{dsyrk}, the classical cubic-time routine
over a finite field (calling \texttt{dsyrk} and performing modular
reductions on the result), and the classical divide and conquer
algorithm~\cite[\S~6.3.1]{jgd:2008:toms}.

\begin{figure}[htbp]
  \begin{center}
    \includegraphics[width=\columnwidth]{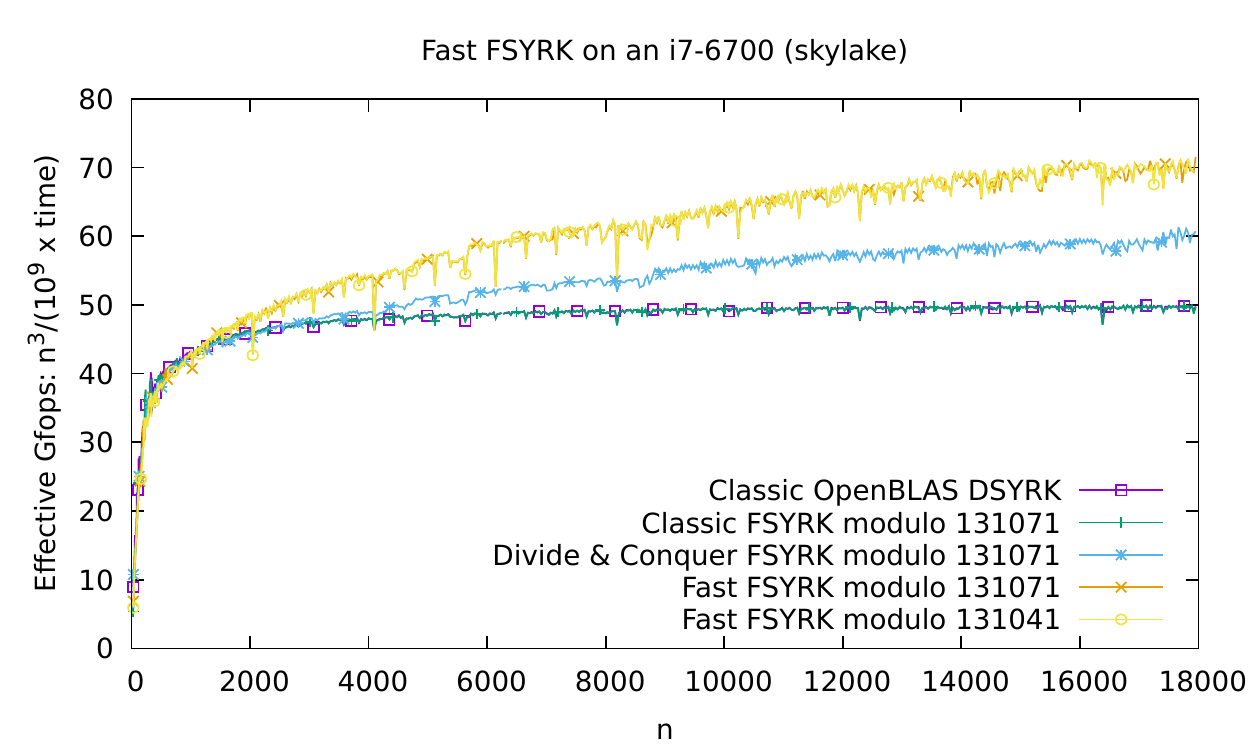}
    \caption{Speed of an implementation of Algorithm~\ref{alg:MIAHMM}}\label{fig:perfs}%
  \end{center}
\end{figure}

The \texttt{fflas-ffpack} library is linked with
Open\textsc{blas}~\cite[v0.3.6]{openblas}
and compiled with \texttt{gcc-9.2} on an Intel skylake~i7-6700 running
a Debian \textsc{gnu}/Linux system (v5.2.17).
\par

The slight overhead of performing the modular reductions is quickly
compensated by the speed-up of the sub-cubic algorithm (the threshold
for a first recursive call is near~${n=2000}$).  The classical divide
and conquer approach also speeds up the classical algorithm, but
starting from a larger threshold, and hence at a slower pace.  Lastly,
the speed is merely identical modulo~${131041}$, where square roots
of~$-1$ exist, thus showing the limited overhead of the
preconditioning by the matrix~$Y$.

\section{Perspective}

We made progresses in order to prove that five non-commutative products are necessary for the computation
of the product of a $\matrixsize{2}{2}$ matrix by its adjoint, by applying de Groote's method to
this context. However, we only prove that there is no algorithm,
derived from a bilinear one, which uses 4 products and the
adjoint of one of them. The case where more adjoints of already
computed products could be used need to be ruled out in a similar manner.
More generally, the possible existence of algorithms not originating
from a bilinear algorithm is an even more challenging question.

Over the algebra of quaternions, the natural generalization of Howell
and Lafon's algorithm to matrix coefficients yields the 7M and 6M
algorithms for the transpose and conjugate transpose respectively.
The recursive 5 products algorithm is only usable for the conjugate
transpose case in positive characteristic and costs, as expected, half
the cost of a general quaternion matrix product for
$\omega=\log_2{7}$.
Yet for $\omega<2.536$ and for the case of transposition the 6M and 7M
algorithms perform best.
The minimality of the number $6$ of multiplications to multiply a
quaternion by its conjugate is an open question, as for the minimality
of the number $7$ of multiplications to multiply a quaternion by its
transpose.
For these questions, de Groote's method could provide an answer.

We proposed several algorithms for the product of a matrix by its
adjoint, each of which improves by a constant factor the best known
costs, depending on the algebraic nature of the field of coefficients
and on the underlying matrix exponent to be chosen.
When implemented in practice the comparison may become even more
complex, as other parameters, such as memory access pattern or
vectorization will come into play.
Our first experiments show that these constant factor improvements do
have a practical impact.

\bibliography{quat,strassen,RFC2009}

\def\noopsort#1{}\newcommand{\issacproceedings}[2]{{ISSAC}'#1, Proceedings of
  the #1 International Symposium on Symbolic and Algebraic Computation, #2}
\begin{thebibliography}{10}

\bibitem{Baboulin:2005:csyrk}
M.~Baboulin, L.~Giraud, and S.~Gratton.
\newblock A parallel distributed solver for large dense symmetric systems:
  Applications to geodesy and electromagnetism problems.
\newblock {\em Int. J. of HPC Applications}, 19(4):353--363, 2005.
\newblock \href {https://doi.org/10.1177/1094342005056134}
  {\path{doi:10.1177/1094342005056134}}.

\bibitem{Beniamini:2019:fmmsd}
G.~Beniamini and O.~Schwartz.
\newblock Faster matrix multiplication via sparse decomposition.
\newblock In {\em Proc. SPAA'19}, pages 11--22, 2019.
\newblock \href {https://doi.org/10.1145/3323165.3323188}
  {\path{doi:10.1145/3323165.3323188}}.

\bibitem{Dumas:2009:WinoSchedule}
Brice Boyer, Jean-Guillaume Dumas, Cl\'ement Pernet, and Wei Zhou.
\newblock Memory efficient scheduling of {Strassen-Winograd}'s matrix
  multiplication algorithm.
\newblock In {\em Proc.}, ISSAC'09, pages 135--143. ACM Press, July 2009.
\newblock \href {https://doi.org/10.1145/1576702.1576713}
  {\path{doi:10.1145/1576702.1576713}}.

\bibitem{brillhart:1972:twosquares}
J.~Brillhart.
\newblock Note on representing a prime as a sum of two squares.
\newblock {\em Math. of Computation}, 26(120):1011--1013, 1972.
\newblock \href {https://doi.org/10.1090/S0025-5718-1972-0314745-6}
  {\path{doi:10.1090/S0025-5718-1972-0314745-6}}.

\bibitem{bshouty:1995a}
N.~H. Bshouty.
\newblock On the additive complexity of~{${2 \times 2}$} matrix multiplication.
\newblock {\em Inf. Processing Letters}, 56(6):329--335, December 1995.
\newblock \href {https://doi.org/10.1016/0020-0190(95)00176-X}
  {\path{doi:10.1016/0020-0190(95)00176-X}}.

\bibitem{jgd:2008:toms}
J.-G. Dumas, P.~Giorgi, and C.~Pernet.
\newblock Dense linear algebra over prime fields.
\newblock {\em ACM TOMS}, 35(3):1--42, November 2008.
\newblock \href {https://doi.org/10.1145/1391989.1391992}
  {\path{doi:10.1145/1391989.1391992}}.

\bibitem{jgd:2020:wishart}
Jean-Guillaume Dumas, Cl\'ement Pernet, and Alexandre Sedoglavic.
\newblock On fast multiplication of a matrix by its transpose.
\newblock In {\em Proc.}, ISSAC'20, pages 162--169, New York, July 2020. ACM
  Press.
\newblock \href {https://doi.org/10.1145/3373207.3404021}
  {\path{doi:10.1145/3373207.3404021}}.

\bibitem{fflas19}
{\noopsort{FFLAS-FFPACK}}{The FFLAS-FFPACK group}.
\newblock {\em {FFLAS-FFPACK}: {F}inite {F}ield {L}inear {A}lgebra
  {S}ubroutines / {P}ackage}, 2019.
\newblock v2.4.1.
\newblock URL: \url{http://github.com/linbox-team/fflas-ffpack}.

\bibitem{Fiduccia:1971:fmm}
Charles~M. Fiduccia.
\newblock Fast matrix multiplication.
\newblock In {\em Proc.}, STOC '71, pages 45--49, New York, NY, USA, 1971. ACM
  Press.
\newblock \href {https://doi.org/10.1145/800157.805037}
  {\path{doi:10.1145/800157.805037}}.

\bibitem{groote:1978}
Hans~Friedich Groote, de.
\newblock On varieties of optimal algorithms for the computation of bilinear
  mappings {II}. {O}ptimal algorithms for~{${2\times 2}$}-matrix
  multiplication.
\newblock {\em Theoretical Computer Science}, 7(2):127--148, 1978.
\newblock \href {https://doi.org/10.1016/0304-3975(78)90045-2}
  {\path{doi:10.1016/0304-3975(78)90045-2}}.

\bibitem{Groote:1975:quaternion}
Hans~Friedrich Groote, de.
\newblock On the complexity of quaternion multiplication.
\newblock {\em Inf. Processing Letters}, 3(6):177 -- 179, 1975.
\newblock \href {https://doi.org/10.1016/0020-0190(75)90036-8}
  {\path{doi:10.1016/0020-0190(75)90036-8}}.

\bibitem{groote:1978a}
Hans~Friedrich Groote, de.
\newblock On varieties of optimal algorithms for the computation of bilinear
  mappings {I}. {T}he isotropy group of a bilinear mapping.
\newblock {\em Theoretical Computer Science}, 7(2):1--24, 1978.
\newblock \href {https://doi.org/10.1016/0304-3975(78)90038-5}
  {\path{doi:10.1016/0304-3975(78)90038-5}}.

\bibitem{Higham:1992:complex3M}
N.~J. Higham.
\newblock Stability of a method for multiplying complex matrices with three
  real matrix multiplications.
\newblock {\em SIMAX}, 13(3):681--687, 1992.
\newblock \href {https://doi.org/10.1137/0613043} {\path{doi:10.1137/0613043}}.

\bibitem{Howell:1975:quatprodeight}
Thomas~D. Howell and Jean Lafon.
\newblock The complexity of the quaternion product.
\newblock Technical report, Cornell University, USA, 1975.
\newblock URL: \url{https://hdl.handle.net/1813/6458}.

\bibitem{Karstadt:2017:strassen}
E.~Karstadt and O.~Schwartz.
\newblock Matrix multiplication, a little faster.
\newblock In {\em Proc. SPAA'17}, pages 101--110. ACM, 2017.
\newblock \href {https://doi.org/10.1145/3087556.3087579}
  {\path{doi:10.1145/3087556.3087579}}.

\bibitem{Landsberg:2016ab}
Joseph~M. Landsberg.
\newblock {\em Geometry and complexity theory}, volume 169 of {\em Cambridge
  Studies in Advanced Mathematics}.
\newblock Cambrigde University Press, December 2016.
\newblock \href {https://doi.org/10.1017/9781108183192}
  {\path{doi:10.1017/9781108183192}}.

\bibitem{LeGall:2014:fmm}
F.~Le~Gall.
\newblock Powers of tensors and fast matrix multiplication.
\newblock In {\em Proc ISSAC'14}, pages 296--303. ACM, 2014.
\newblock \href {https://doi.org/10.1145/2608628.2608664}
  {\path{doi:10.1145/2608628.2608664}}.

\bibitem{Seroussi:1980:BBgfp}
G.~Seroussi and A.~Lempel.
\newblock Factorization of symmetric matrices and trace-orthogonal bases in
  finite fields.
\newblock {\em SIAM J. on Computing}, 9(4):758--767, 1980.
\newblock \href {https://doi.org/10.1137/0209059} {\path{doi:10.1137/0209059}}.

\bibitem{Strassen:1969:GENO}
V.~Strassen.
\newblock {G}aussian elimination is not optimal.
\newblock {\em Numerische Mathematik}, 13:354--356, 1969.
\newblock \href {https://doi.org/10.1007/BF02165411}
  {\path{doi:10.1007/BF02165411}}.

\bibitem{Wedeniwski:2001:lqnr}
S.~Wedeniwski.
\newblock Primality tests on commutator curves.
\newblock PhD U. T{\"u}bingen, 2001.
\newblock URL: \url{https://d-nb.info/963295438/34}.

\bibitem{Winograd:1977:complexite}
S.~Winograd.
\newblock La complexit{\'e} des calculs num{\'e}riques.
\newblock {\em La Recherche}, 8:956--963, 1977.

\bibitem{openblas}
Zhang Xianyi, Martin Kroeker, et~al.
\newblock {\em {OpenBLAS}, an Optimized {BLAS} library}, 2019.
\newblock \url{http://www.openblas.net/}.

\end{thebibliography}

\end{document}